\documentclass[usletter,11pt]{amsart}
\usepackage[utf8]{inputenc}
\usepackage{amsmath}
\usepackage{amsfonts}
\usepackage{amsthm}
\usepackage{amssymb}
\usepackage{setspace}
\usepackage[colorinlistoftodos]{todonotes}

\numberwithin{equation}{section}
\theoremstyle{plain}
\newtheorem{thm}{Theorem}[section]
\newtheorem{lem}[thm]{Lemma}
\newtheorem{prop}[thm]{Proposition}
\newtheorem{cor}[thm]{Corollary}
\theoremstyle{definition}
\newtheorem{defn}[thm]{Definition}
\newtheorem{rem}[thm]{Remark}

\usepackage[marginparwidth = 2.3cm,hmarginratio=1:1]{geometry}

\usepackage[hidelinks]{hyperref}

\title[Sparse recovery in bounded Riesz systems and numerical methods for PDEs]{Sparse 
recovery in bounded
Riesz systems with applications to numerical methods for PDEs}
\author[S. Brugiapaglia]{Simone Brugiapaglia}
\address{Department of Mathematics and Statistics, Concordia University. Montreal, QC, Canada.}
\email{simone.brugiapaglia@concordia.ca}
\author[S. Dirksen]{Sjoerd Dirksen}
\address{Mathematical Institute, Utrecht University. Utrecht, The Netherlands.}
\email{s.dirksen@uu.nl}
\author[H.C. Jung]{Hans Christian Jung}
\address{Chair for Mathematics of Information Processing, RWTH Aachen University. Aachen, Germany.}
\email{jung@mathc.rwth-aachen.de}
\author[H. Rauhut]{Holger Rauhut}
\email{rauhut@mathc.rwth-aachen.de}

\def\ep{\varepsilon}

\def\cbb{\mathbb{C}}
\def\ebb{\mathbb{E}}

\def\pbb{\mathbb{P}}

\def\rbb{\mathbb{R}}

\def\A{\mathcal{A}}
\def\B{\mathcal{B}}

\def\F{\mathcal{F}}

\def\I{\mathcal{I}}

\def\N{\mathcal{N}}

\def\X{\mathbf{X}}

\def\1{\mathbf{1}}

\def\supp{\operatorname{supp}}

\def\sparse{\Sigma_{s,N}}

\newcommand{\norm}[2]{\Vert #1 \Vert_{#2}}

\newcommand{\conv}[1]{\operatorname{conv}(#1)}
\newcommand{\scalar}[2]{\langle #1, #2 \rangle}
\newcommand{\inner}[2]{\langle #1, #2 \rangle}

\newcommand{\useequationref}[1]{\refstepcounter{equation} \label{#1} (\theequation)}

\usepackage{algorithmic}
\usepackage{enumerate}

\newtheorem{algo}[thm]{Algorithm}

\DeclareMathOperator{\Fop}{\mathcal{F}}

\DeclareMathOperator{\Span}{span}
\DeclareMathOperator{\BigO}{\mathcal{O}}

\newcommand{\RR}{\mathbb{R}}

\newcommand{\NN}{\mathbb{N}}

\newcommand{\de}{\,\text{d}}

\newcommand{\Prob}{\mathbb{P}}
\newcommand{\pr}{\mathbb{P}}
\newcommand{\Expe}{\mathbb{E}}

\newcommand{\RIP}[2]{\operatorname{RIP}({#1},{#2})}

\newcommand{\deltaM}{\gamma}

\newcommand{\condweak}{\kappa}

\newcommand{\OMPerr}{C}
\newcommand{\normubound}{L}

\allowdisplaybreaks

\begin{document}

\maketitle
\begin{abstract}
We study sparse recovery with structured random measurement matrices having independent, identically distributed, and uniformly bounded rows and with a nontrivial covariance structure. This class of matrices arises from random sampling of bounded Riesz systems and generalizes random partial Fourier matrices. Our main result improves the currently available results for the null space and restricted isometry properties of such random matrices.  The main novelty of our analysis is a new upper bound for the expectation of the supremum of a Bernoulli process associated with a restricted isometry constant. 
We apply our result to prove new performance guarantees for the \textsf{CORSING} method, a recently introduced numerical approximation technique for partial differential equations (PDEs) based on compressive sensing.
\end{abstract}


\section{Introduction}

Compressive sensing \cite{candes_robust_2006,candes_decoding_2005,donoho_compressed_2006,holger,rudelson_sparse_2008} provides efficient
methods that allow to recover (approximately) sparse vectors from a surprisingly small amount of random measurements. 
Although compressive sensing was initially motivated by signal processing applications, it has recently inspired a new generation of hybrid methodologies in computational mathematics. This includes compressive sensing techniques for polynomial interpolation \cite{adcock2018infinite,rauhut_random_2007,rauhut2012sparse}, high-dimensional function approximation \cite{adcock2019correcting,adcock2017compressed,webster}, the numerical solution of PDEs \cite{Brugiapaglia2015,jokar2010sparse,kang2019economical}, uncertainty quantification of PDEs with random inputs \cite{doostan2011non,mathelin2012compressed,rauhut2017compressive,yan2012stochastic,yang2013reweighted},  dynamical systems \cite{tran2017exact}, and inverse problems in PDEs \cite{alberti2019infinite}. 
In this paper, our main application of interest is the numerical approximation of PDEs based on compressive sensing via the \textsf{CORSING} (\textsf{COmpRessed SolvING}) method \cite{Brugiapaglia2015,brugiapaglia2018theoretical}.

A key theoretical task in compressive sensing is to provide estimates for the so-called restricted isometry constants or, strongly related to this, to derive the null space property,
of random measurement matrices. Such estimates lead to bounds on the sufficient number of measurements for recovery in terms of
the sparsity and the dimension of the vector to be recovered. We also note that, besides compressive sensing, the estimation of the restricted isometry constant of a matrix is a central 
problem in, e.g., the analysis of list-decodeable linear codes \cite{cheraghchi_restricted_2013}, and Johnson-Lindenstrauss embeddings \cite{krahmer_new_2011,ORS18}.

A class of structured measurement matrices, which plays an important role in compressive sensing as well as in the \textsf{CORSING} method for the numerical solution of PDEs, arises from random sampling function systems (such as the Fourier system). 
While corresponding bounds on the restricted isometry constants exist for orthonormal systems that are bounded in the $L^\infty$-norm \cite{candes_robust_2006,holger,rudelson_sparse_2008,bourgain2011}, similar estimates for the more general class of bounded Riesz systems (where orthogonality does not necessary hold) are not yet available in the literature up to the best of the authors' knowledge. The \textsf{CORSING} method, however, typically requires to work with Riesz systems rather than orthonormal systems, which raises the need for such a generalization.

The contribution of this paper is twofold. On the one hand, we provide a new analysis of the restricted isometry constants and null space property of
matrices arising from random sampling in bounded Riesz systems, which also improves 
the available bounds for the orthonormal case. On the other hand, we take advantage of this result to obtain substantially improved theoretical guarantees for the \textsf{CORSING} method. 

\subsection{Main results}

Recall that the 
restricted isometry constant of sparsity level $s$ of a matrix $A$ with $N$ columns is 
defined by 
$$
\ep_s := \sup_{f\in D_{s,N} } \Big| \|Af\|_2^2 - 1\Big|,
$$
where $D_{s,N} := \{f \in \cbb^N : \norm{f}{0} \leq s, \norm{f}{2} = 1\}$ is the 
set of unit-norm vectors with support size at most $s$. Our main theorem establishes a concentration inequality, which implies bounds for the restricted 
isometry constants (and establishes the null space property) of a random matrix whose rows are independent, identically 
distributed, and uniformly bounded random vectors.
\begin{thm}
\label{prop:main}
There exist absolute constants $\kappa, c_0, c_1 >0$ such that the following holds.
Let $X_1,\ldots,X_m$ be independent copies of a random 
vector $X \in \cbb^N$ with bounded coordinates, i.e., for all $i = 1,\ldots,N$ we have 
$|\inner{X}{e_i}|\leq K$ for some $K>0$ where $e_1,\ldots,e_N$ is the standard basis of $\mathbb{C}^N$. 
Let $T\subseteq \{ f \in \cbb^N :\norm{f}{1}
\leq \sqrt{s} \}$, $\delta \in (0,\kappa)$ and assume that
\begin{equation}
\label{eqn:numSampMain}
m \geq c_0 \, K^2 \delta^{-2} s \log(e N) \log^2(sK^2/\delta) 
 \; .
\end{equation}
Then, with probability exceeding $1-2\exp(-\delta^{2} m /(sK^2))$, 
\begin{equation}
	\label{eq:main_result}
 \sup_{f \in T} \Big|\frac{1}{m}\sum_{i=1}^m 
|\inner{f}{X_i}|^2 - \ebb |\inner{f}{X}|^2 \Big| \leq c_1 \Big(\delta + 
\delta \sup_{f \in T} \ebb|\inner{f}{X}|^2 \Big)\; .
\end{equation}
\end{thm} 
\begin{rem}
    The values of the constants in the above theorem can be chosen as $\kappa = \frac{1}{28}(10-7\sqrt{2}) \approx 0.3066$, 
    $c_0 = 1600(99+70\sqrt{2}) \approx 316792$ and $c_1 = 492$. We have not optimized these constants in our proof and believe that they can be further reduced.
\end{rem}
Theorem~\ref{prop:main} fits into a line of 
work \cite{bou:improved_rip,candes_near_2006,webster,DBLP:journals/corr/HavivR15a,rudelson_sparse_2008} that studies the restricted isometry constants of random matrices associated with random sampling from a bounded orthonormal system. 
In the breakthrough work \cite{candes_near_2006} it was established
that 
the restricted isometry constant $\ep_s$
of such a matrix satisfies $\ep_s \leq \delta_0$ with
probability at 
least $1-N^{-c_0 \tau}$ provided that the number of measurements $m$ satisfies $m \geq c_1 \tau s \log^6(N)$ (where $c_0$ and $c_1$ only depend on $\delta_0$).
In \cite{rudelson_sparse_2008} generic chaining 
techniques were used to 
show that $m\geq c_0 
\ep^{-2} s \log(s \log(N)) \log^2(s) \log(N)$ suffice to guarantee that $\ep_s\leq 
\ep$ with high probability. Using a different technique, \cite{bou:improved_rip} 
showed that $m\geq c_0 \ep^{-6} s\log(s)\log^2(N)$ measurements are sufficient for a discrete 
bounded orthonormal system. The latter result was generalized in \cite{webster} 
to continuous sampling scenarios. Inspired by the methods in \cite{bou:improved_rip}
it was shown in \cite{DBLP:journals/corr/HavivR15a} 
that for a discrete bounded orthonormal system $m \geq c_0 \ep^{-2} \log^2(s/\ep) 
\log(eN) \log(1/\ep)$ measurements are sufficient. Recently, the work \cite{blasiok_RIP_2019} has shown that the condition $m \geq c_0 s \log(s)
\log(N/s)$ is necessary for a randomly subsampled Hadamard matrix to satisfy $\ep_s\leq c_1$.
Hence, for constant $\delta$, $K$ and for $s \leq N^{\alpha}$ for some $\alpha < 1$, say, the bound \eqref{eqn:numSampMain} is optimal up to a factor of $\log(s)$.

Theorem~\ref{prop:main} improves on these state-of-the-art results in several ways. In the setting of continuous bounded orthonormal systems, it shows that $m \geq c_0 \ep^{-2} s \log(e N) \log^2(s/\ep)$ measurements suffice to guarantee that $\ep_s\leq 
\ep$ with high probability. In particular, compared to 
\cite{DBLP:journals/corr/HavivR15a}, we remove a $\log(1/\ep)$ factor and achieve 
a result in the setting of continuous sampling. Compared to \cite{webster}, we 
improve the dependence on $\ep$. Finally, Theorem~\ref{prop:main} more generally yields estimates on the restricted isometry constants of matrices associated with random sampling from a bounded Riesz system. We refer to Theorem~\ref{thm:riesz_rip} and Remark~\ref{rem:riesz_rip} for details.

Our proof of Theorem~\ref{prop:main} is inspired by the methods in \cite{bou:improved_rip} and \cite{DBLP:journals/corr/HavivR15a}, but in contrast to these references we develop our proof in the terminology of generic chaining \cite{talagrand_upper_2014}.
We believe that this makes the proof more transparent, at least, for those familiar with generic chaining techniques. Moreover, our proof leads to the improved dependence in $\delta$ in \eqref{eqn:numSampMain} (resp.~in $\ep$ in estimates for the RIP). 

Although Theorem~\ref{prop:main} is 
concerned with subsets of
the $\ell^1$-ball, it is possible to extend it
to subsets of a weighted $\ell^1$-ball, i.e., the unit ball of
the norm $\norm{f}{\omega,1} := \sum_{j=1}^n \omega_j |f_j|$,
where $\omega_j \geq 1$ is a sequence of weights. We refer the reader to Section~\ref{sec:weighted_l1} for more details on
this extension and its applications to weighted compressive sensing and uncertainty quantification.

We apply Theorem~\ref{prop:main} to obtain improved robust recovery guarantees for the \textsf{CORSING} method. This is a recently introduced numerical method for computing a sparse approximation of the solution to a PDE based on compressive sensing 
\cite{simone_thesis,Brugiapaglia2015,brugiapaglia2018theoretical}. Given a PDE 
admitting a weak formulation, the idea of the method is to assemble a reduced 
Petrov-Galerkin discretization via random sampling and to solve this reduced 
system using a sparse recovery algorithm. Compared to other nonlinear approximation 
methods for PDEs, such as adaptive finite elements or adaptive wavelets methods 
(see, e.g., \cite{urban2009wavelet} and references therein), \textsf{CORSING} has the 
advantages that no \emph{a posteriori} error indicators are needed and that the 
assembly of the discretization matrix 
as well as the sparse recovery step (here performed via Orthogonal Matching Pursuit 
(OMP)) can be easily parallelized. We refer to \cite{brugiapaglia2018theoretical} 
for a more detailed discussion and to 
\cite{simone_thesis,Brugiapaglia2018wavelet,Brugiapaglia2015,brugiapaglia2018theoretical} 
for numerical experiments for multi-dimensional advection-diffusion-reaction 
equations and the Stokes problem.

The best available theoretical guarantees for the \textsf{CORSING} method \cite{brugiapaglia2018theoretical} state that in order to 
recover the best $s$-term approximation of the PDE solution with respect to a Riesz basis
of trial functions, 
it is sufficient to assemble a number
of rows proportional to $s^2$ (up to logarithmic factors) under suitable assumptions 
involving the trial and test functions and the bilinear form defining the weak formulation of the PDE. 
This is highly suboptimal compared to standard compressive
sensing results, where a number of measurements proportional to $s$ (up to log factors) 
is usually sufficient. Indeed, numerical experiments show that the quadratic scaling $s^2$ is highly 
pessimistic (see \cite[Figure 8]{brugiapaglia2018theoretical}). Another limitation of the state-of-the-art results is the assumption that one solves an NP-hard problem \emph{exactly} in the recovery phase.

In this paper, we bridge
these gaps by showing robust recovery guarantees under optimal linear scaling between $m$ and $s$ (up to log factors) which cover sparse recovery via $\ell^1$-minimization and via orthogonal matching pursuit (OMP). The latter is preferred in practice. This result is stated in Theorem~\ref{thm:CORSINGOMPrecovery}.
The main technical challenge is to establish an 
improved bound on the sufficient size of the reduced Petrov-Galerkin discretization 
by analyzing the restricted isometry constant of the corresponding random matrix, which 
does not have a trivial covariance. This challenge is overcome thanks to Theorem~\ref{prop:main}.

\subsection{Organization of the paper}
In Section~\ref{sec:results} we present several consequences of Theorem~\ref{prop:main} 
for sparse recovery via $\ell^1$-minimization in the case of subsampled bounded 
Riesz systems (Sections~\ref{sec:Riesz} and \ref{sec:finite_Riesz}). Taking advantage 
of the theory presented in Section~\ref{sec:results}, we prove new recovery guarantees 
for the \textsf{CORSING} method in Section~\ref{sec:PDEs}, focusing on the case where 
recovery is performed via orthogonal matching pursuit. Section~\ref{section/proof} 
develops the proof of Theorem~\ref{prop:main} and is the technical core of 
the paper.

\subsection{Notation}
Throughout the paper we will use the following notation.
Given $N \in \mathbb{N}$, we define $[N]:=\{1,\ldots,N\}$. Moreover, given a 
vector $x \in \mathbb{C}^N$, we denote its $\ell^p$-norm as $\|x\|_p = (\sum_{j 
= 1}^N|x_j|^p)^{1/p}$ for $p > 0$ and its $\ell^0$-norm as $\|x\|_0 = 
|\{j\in[N]: x_j \neq 0\}|$. Moreover, we let $B_{\ell^p}^N = \{x \in \mathbb{C}^N: \|x\|_p \leq 1\}$ for $p >0$.  We denote the set of $s$-sparse vectors of $\mathbb{C}^N$ as $\Sigma_{s,N} := \{x \in \cbb^N : \norm{x}{0} \leq s\}$ and the set of $s$-sparse unit vectors of $\mathbb{C}^N$ 
as $D_{s,N} := \{x \in \cbb^N : \norm{x}{0} \leq s, \norm{x}{2} = 1\}$. Finally, 
the notation $X \lesssim Y$ hides the presence of a constant $c >0$ independent of 
$X$ and $Y$ such that $X \leq c Y$. Given a set $S$, $|S|$ denotes its cardinality. Moreover, if $S \subseteq [N]$, then $S^c$ is the complement set of $S$ with respect to $[N]$. We use letters $c,c'$ or $c_k$ with $k \in \mathbb{N}$ to denote universal (or absolute) constants and it is understood that such numbers are not necessarily the same on every occurrence.

\section{Sparse recovery in bounded Riesz systems}
\label{sec:results}
This section outlines several applications of our main result,
Theorem~\ref{prop:main}, in compressive sensing. We focus on new sparse
recovery guarantees for subsampled bounded Riesz systems,
which extend previously known results for these settings 
(Sections~\ref{sec:Riesz} and \ref{sec:finite_Riesz}). Most of the recovery
guarantees rely on the notion of restricted isometry property.
\begin{defn}[Restricted isometry property]
	The restricted isometry constant $\ep_s$ of a matrix $A \in \cbb^{m
	\times N}$ is defined by
	\begin{equation}
		\ep_s := \sup_{f\in D_{s,N}} \Big| \|Af\|_2^2 - 1\Big|.
		\label{eq:definition_rip}
	\end{equation}
	If a matrix $A \in \cbb^{m \times N}$ satisfies 
	$\ep_s \leq \ep$ for some $\ep>0$, then 
	we say that $A$ satisfies $\RIP{\ep}{s}$. 
\end{defn}

\subsection{Subsampled bounded Riesz systems}
\label{sec:Riesz}
The first application of Theorem~\ref{prop:main} is concerned with the recovery of functions having a sparse 
expansion with respect to a Riesz system. We start by recalling 
the definition of a Riesz system.
\begin{defn}[Riesz system]
Let $(H, \scalar{\cdot}{\cdot}_H)$ denote a complex Hilbert space. A sequence
$(\psi_j)_{j \in \mathbb{N}}$ with $\psi_j \in H$ is called a Riesz system if 
there exist constants $ 0< c_\psi \leq C_\psi < \infty$, such that for every $f 
= (f_n)_{n \in \mathbb{N}} \in \ell^2(\cbb)$,
\begin{equation} 
\label{def/riesz}
 c_\psi \norm{f}{\ell^2(\cbb)}^2 
\leq \Big \Vert \sum_{j \in \mathbb{N}} f_j
\psi_j \Big \Vert_H^2 \leq C_\psi \norm{f}{\ell^2(\cbb)}^2.
\end{equation}
\end{defn}
Let us point out that an orthogonal system is a Riesz system with constants $c_\psi
= C_\psi =1$.

In the following, we focus on Riesz systems for Hilbert spaces 
of the form $H = L^2(S, \mu)$, where $S \subseteq \rbb^r$ and 
$\mu$ is a probability measure on $S$. For the Hilbert space $L^2(S,\mu)$ the inner product
$\inner{\cdot}{\cdot}_{L^2(S,\mu)}$ and the corresponding norm $\norm{\cdot}{L^2(S,\mu)}$ are
given by 
\begin{equation*}
    \inner{f}{g}_{L^2(S,\mu)} = \int_S f(\omega) \overline{g(\omega)} \; d\mu(\omega)
    \quad \text{and} \quad
    \norm{f}{L^2(S,\mu)} = \sqrt{\inner{f}{g}_{L^2(S,\mu)}} \; .
\end{equation*}
We consider a sequence $(\psi_j :
S \to \cbb)_{j \in \mathbb{N}}$ of bounded, measurable functions satisfying
\eqref{def/riesz}. Let $F \in H$ be a function with
a finite and sparse Riesz expansion, i.e.,
\begin{equation}
\label{eq:riesz/expansion}
F = \sum_{j \in [N]} f_j \psi_j \quad \text{and} \quad f \in \sparse\; .
\end{equation} 
In this setting, we are interested in the problem of recovering the coefficients
$f$ of the function $F \in L^2(S,\mu)$ with respect to $(\psi_j)_{j \in [N]}$ from a 
finite number of samples $F(\omega_1), \ldots F(\omega_m)$, where $\omega_1, \ldots, \omega_m \in S$ are drawn i.i.d.\ at random
with respect to the measure $\mu$. 
We will outline two approaches to address this question. The first one is based on the restricted isometry property (Theorem~\ref{thm:riesz_rip}) and the second one is based on the $\ell^2$-robust null space property  (Theorem~\ref{thm:riesz_stable_recovery}). 

The first approach is to show that the matrix
\begin{equation} 
\label{eq:riez/measurement}
 A = \frac{1}{\sqrt{m C_\psi}} (\psi_j(\omega_i))_{i \in [m], j \in [N]},
\end{equation}
satisfies RIP$(\ep, s)$ under suitable conditions on $m$, $s$, and 
$\varepsilon$.
If $A$ satisfies RIP$(\ep,2s)$ for $\ep \leq 1/\sqrt{2}$, then the
coefficients of $F$ with respect to $(\psi_j)_{j \in [N]}$ can be
recovered from noisy observations $y_1=F(\omega_1) + e_1, \ldots, y_m=F(\omega_m) + e_m$
with $\norm{e}{2} \leq \zeta$ via the quadratically-constrained basis pursuit program (see, e.g., 
\cite[Theorem 2.1]{cai_sparse_2014}) 
\begin{equation}
 	\label{def/program/eta}
	\min_{z \in \cbb^N} \norm{z}{1} \qquad \text{subject to} \qquad 
	\norm{ \sqrt{C_\psi m} Az - y}{2} \leq \zeta \; .
	\tag{BP$_\zeta$}
\end{equation} 
This approach will lead to some limitations 
on the restricted isometry constant $\ep_s$ that are particularly restrictive when the ratio $c_\psi/C_\psi$ is small. 
\begin{thm} 
\label{thm:riesz_rip}
There are absolute constants $c_0,c_1>0$, such that the following holds.
Let $H = L^2(S, \mu)$ and let $(\psi_j)_{j \in [N]}$ be a 
Riesz system. Let $\varepsilon \in (1-\frac{c_\psi}{C_\psi},1)$, $K_\psi =
\max_{j\in [N]}  \norm{\psi_j}{L^\infty(S)}$ and assume that
\begin{equation}
\label{eq:riesz/m_bound}
m \geq 
c_0 \max\{C_\psi^{-2},1\} 
\, K_\psi^2 \eta^{-2}  \, 
s \log^2(s K_\psi^2\max\{C_\psi^{-2},1\} \eta^{-2}) \log(eN) \; ,
\end{equation}
where $\eta = (\ep-1+\frac{c_\psi}{C_\psi}) >0$. Then, with probability 
at least $$1-2\exp(- c_1 \min\{C_\psi^2,1\} \eta^2  m / ( K_\psi^2 s))$$ 
the matrix $A$ defined in \eqref{eq:riez/measurement} satisfies $\RIP{\varepsilon}{s}$.
\end{thm}
\begin{rem} \label{rem:riesz_rip}
	As pointed out above, an orthonormal system is a special case of a Riesz 
	system with $c_\psi = C_\psi = 1$. 
	Therefore, 
	Theorem~\ref{thm:riesz_rip} also recovers and extends known results for bounded
	orthonormal systems \cite{bou:improved_rip,DBLP:journals/corr/HavivR15a,webster}.
%
\end{rem}
\begin{proof}[Proof of Theorem \ref{thm:riesz_rip}]
Let $X(\omega):=(\psi_j(\omega))_{j \in[N]}$, where $\omega \in S$
is chosen at random with respect to $\mu$, and let $X_1, \ldots, X_m$ denote
independent copies of $X$. Then, recalling 
\eqref{def/riesz}, \eqref{eq:riez/measurement} and that $\norm{f}{2}=1$
for $f \in D_{s,N}$ we see that 
$$
\ebb|\scalar{X}{f}|^2 = \|F\|_{L^2(S,\mu)}^2 \in [c_\psi \norm{f}{2}^2, 
C_\psi \norm{f}{2}^2] = [c_\psi,C_\psi]
$$ 
and 
$$
\frac{1}{m} \sum_{j = 1}^m |\scalar{X_i}{f}|^2 = C_\psi \|Af\|_2^2.
$$
Further,
let $\kappa, c_1>0$ denote the constants in Theorem~\ref{prop:main} and observe that $D_{s,N} \subseteq \sqrt{s} B_{\ell^1}^N$. 
Therefore, setting $\eta = \ep - 1 + \frac{c_\psi}{C_\psi}$ and applying Theorem~\ref{prop:main} with $\delta = \frac{1}{2} c_1 \min\{C_\psi, 1\} 
\cdot \eta $, $T=\sqrt{s} B_{\ell^1}^N \cap S^{N-1}$ and noticing
that with these parameters we have $\delta < \kappa$, 
shows that the following event occurs with probability exceeding 
$1-2 \exp( - c \min\{C_\psi^2, 1\} \eta^2 m / ( sK_{\psi}^2))$: 
\begin{align*}
	\sup_{f \in D_{s,N}} \Big| \norm{A f}{2}^2 - \frac{1}{C_\psi}\ebb|\inner{f}{X}|^2\Big|
	&\leq \frac{\min\{C_\psi,1\}}{C_\psi} \frac{\eta}{2} + \frac{\min\{C_\psi,1\}}{C_\psi}
	\frac{\eta}{2} \sup_{f\in D_{s,N}} \ebb |\inner{f}{X}|^2 \\
	&\leq \min\{C_\psi,1\} \Big(1 + \frac{1}{C_\psi} \Big) \frac{\eta}{2}
	\leq \eta \; .
\end{align*}
provided that \eqref{eq:riesz/m_bound} is satisfied for a suitiable constant $c_0>0$. 
In the above event, we find that for
all $f \in D_{s,N}$ the inequality 
\begin{equation}
	\label{eq:riesz/event}
	\frac{c_\psi}{C_\psi} - \eta \leq \norm{A f}{2}^2 \leq 1 + \eta,
\end{equation}
holds. Recalling the definition of $\eta$, the inequality 
\eqref{eq:riesz/event} reads
$$
(1-\ep) 
\leq 
\|Af\|_2^2 
\leq 
\left(1 + \ep - 1 +\frac{c_\psi}{C_\psi}\right) 
\leq (1+\ep), 
$$
which implies $\RIP{\ep}{s}$ for $\ep \in (1-\frac{c_\psi}{C_\psi},1)$ and concludes 
the proof.
\end{proof}

\begin{rem}
An inspection of the proof of Theorem~\ref{thm:riesz_rip} reveals that a 
sufficient assumption on $(\psi_j)_{j \in[N]}$ is that the relation 
\eqref{def/riesz} holds only for every $f \in \sparse$. This relaxed 
assumption will be used for the analysis in Section~\ref{sec:finite_Riesz} and 
for the application to the \textsf{CORSING} method in 
Section~\ref{sec:PDEs}.
\end{rem}

As pointed out above, the fact that the matrix $A$ defined in
\eqref{eq:riez/measurement} only satisfies 
$\RIP{\ep}{s}$ for $\ep \in(1-\frac{c_\psi}{C_\psi}, 1)$ might cause difficulties 
for recovering the coefficients of the function $F$, since, e.g.,
\cite[Theorem 6.12]{holger} requires that $A$ satisfies $\RIP{\ep}{2s}$ for 
$\ep \leq 1/\sqrt{2}$ in order to recover all $s$-sparse signals via 
\eqref{def/program/eta}. Thus, if the ratio $c_\psi/C_\psi$ satisfies $c_\psi/
C_\psi \leq (\sqrt{2}-1)/\sqrt{2}$, then the assumption on the restricted isometry
constant cannot be satisfied. We note that it is possible to relax 
the assumption on $\varepsilon$ and require that $A$ satisfies RIP$(\ep
,{ts})$ with $\varepsilon < \sqrt{(t - 1)/t}$ with $t \geq 4/3$ (see 
\cite[Theorem 2.1]{cai_sparse_2014}). Following this approach requires that $t \geq C_\psi/c_\psi$ and therefore requires the number of measurements $m$ at rate $C_\psi/c_\psi$.

Below we provide an alternative analysis of the problem of recovering the coefficients of $F$ based on the $\ell^2$-robust 
null space property. This will circumvent the limitations due to the restricted isometry property analysis and leads to a better measurement complexity.

\begin{thm} 
\label{thm:riesz_stable_recovery}
There exists absolute constants $c_0,c_1>0$ such that the following
holds.
Let $s,m,N \in \mathbb{N}$ and $A\in\mathbb{C}^{m \times N}$ be defined as in \eqref{eq:riez/measurement}. Let
$\norm{e}{2} \leq \zeta$, $f \in \cbb^N$ and set $y = Af + e \in \cbb^m$. Assume that  
\begin{equation}
	m \geq c_0 \, \left(\frac{\max\{1,C_\psi\}}{c_\psi} \right)^2
	K_\psi^2\; s \, \log^2(s  K_\psi^2 \max\{1,C_\psi\}/c_\psi) \log(eN) \; .
\end{equation}
Then, with probability at least $1-2 \exp(-c_1(c_\psi/\max\{1,C_\psi\})^2 m / (s K_\psi^2) )$
any minimizer
$f^\#$ of the program \eqref{def/program/eta} satisfies
\begin{align*}
 \norm{f-f^\#}{1} &\leq \frac{9}{2}\min_{z \in \Sigma_{s,N}} \norm{f-z}{1} +
14 \frac{C_\psi}{c_\psi} \sqrt{s} \zeta \; , \\
 \norm{f-f^\#}{2} &\leq \frac{9}{2\sqrt{s}} \min_{z \in \Sigma_{s,N}}
\norm{f-z}{1} + 14 \frac{C_\psi}{c_\psi} \zeta \; .
\end{align*}
\end{thm}
\begin{rem}
If $f \in \Sigma_{s,N}$, then $\min_{z \in \Sigma_{s,N}} \norm{f-z}{1} = 0$
and we conclude
that for universal constants $c,c'>0$ $\norm{f-f^\#}{1} \leq c (C_\psi/c_\psi) 
\sqrt{s} \zeta$ and
$\norm{f-f^\#}{2} \leq c'(C_\psi/c_\psi) \zeta$.
\end{rem}
As mentioned above, the proof of Theorem \ref{thm:riesz_stable_recovery}
relies on the notion of the $\ell^2$-robust null space property and its relation
to sparse recovery via $\ell^1$-minimization.
\begin{defn}
	A matrix $A \in \cbb^{m \times N}$ satisfies the $\ell^2$-robust null 
	space property of order $s$ with constants $\alpha \in (0,1)$ and $\tau >0$ if
	\begin{equation*}
 		\norm{v_S}{2} \leq \frac{\alpha}{\sqrt{s}} \norm{v_{S^c}}{1} 
		+ \tau \norm{Av}{2},
	\end{equation*}
	holds for every $v \in \cbb^N$ and every set $S \subseteq [N]$ with $|S |\leq s$.
\end{defn}
The relation between stable recovery of $f \in \cbb^N$ and the $\ell^2$-robust 
null space property is stated by the following theorem (see, 
e.g., \cite[Theorem 4.22]{holger}).
\begin{thm}
\label{thm:l1_recovery}
Let $A \in \cbb^{m \times N}$ satisfy the $\ell^2$-robust null space
property of order $s$ with constants $\alpha \in (0,1)$ and $\tau>0$. Let
$\norm{e}{2} \leq \zeta$, $f \in \cbb^N$ and set $y = Af + e$. Then any minimizer
$f^\#$ of the program \eqref{def/program/eta} satisfies
\begin{align*}
 \norm{f-f^\#}{1} &\leq c_0\min_{z \in \Sigma_{s,N}} \norm{f-z}{1} + c_1
 \sqrt{s} \zeta \; ,  \\
 \norm{f-f^\#}{2} &\leq \frac{c_0}{\sqrt{s}} \min_{z \in \Sigma_{s,N}}
 \norm{f-z}{1} + c_1 \zeta \; ,
\end{align*}
with $c_0 = \frac{(1+\alpha)^2}{1-\alpha}$ and $c_1 = \frac{(3+\alpha)}{1-\alpha} \tau$.
\end{thm}

\begin{proof}[Proof of Theorem \ref{thm:riesz_stable_recovery}]
	By Theorem~\ref{thm:l1_recovery} it suffices to show that
	$A$ as defined in \eqref{eq:riez/measurement} satisfies the $\ell^2$-robust
	null space property of order $s$ with $\alpha = \frac{1}{2}$ and
	$\tau = \frac{2C_\psi}{c_\psi}$. 
	Define the cone 
	\begin{equation}
 		\label{eq:cone}
 		\mathcal{T}_{\nu,s} := \Big\{f \in \cbb^N \; : \; \forall S 
			\subseteq [N] \text{ with } |S| \leq s : \norm{f_S}{2} \geq 
		\frac{\alpha}{\sqrt{s}} \; \norm{f_{S^c}}{1} \Big\} \; .
	\end{equation}
	It is known (see e.g. \cite{DLR2017,KR2015}) that if
	\begin{displaymath}
 		\inf_{z \in \mathcal{T}_{\nu,s} \cap S^{N-1}} \norm{Af}{2} 
		\geq \frac{1}{\tau},
	\end{displaymath}
	with $S^{N-1}:=\{f \in \mathbb{C}^N : \|f\|_2 = 1\}$, then $A$ satisfies 
	the $\ell^2$-robust null space property with constants $\tau$
	and $\alpha$. Moreover, by \cite[Lemma 3]{DLR2017}, 
	\begin{displaymath}
		 \mathcal{T}_{\alpha,s} \subseteq (2+\alpha^{-1})  \conv{D_{s,N}}
		 \subseteq (2+\alpha^{-1}) \sqrt{s}B_{\ell^1}^N \; .
	\end{displaymath}
	Hence, it is enough to show that
	\begin{displaymath}
		\inf_{z \in 4 \sqrt{s}B_{\ell^1}^N \cap S^{N-1}} \norm{Af}{2} \geq \frac{c_\psi}{2C_\psi}.
	\end{displaymath}
    This lower bound is found by applying 
	Theorem~\ref{prop:main} with a minor modification, which corresponds to working with $s'=16s$ and multiplying the right-hand side of \eqref{eq:main_result} by $16$. Since
	$\ebb \norm{A f}{2}^2 \geq \frac{c_\psi}{C_\psi} \norm{f}{2}^2$, 
	\begin{align*}
		\inf_{z \in \sqrt{s'} B_{\ell^1}^N \cap S^{N-1}}\norm{Az}{2}^2 
		&\geq  \inf_{z \in \sqrt{s'} B_{\ell^1}^N \cap S^{N-1}} 
			\ebb \norm{Az}{2}^2 
		- \sup_{f \in \sqrt{s'} B_{\ell^1}^N \cap S^{N-1}} 
		\Big| \norm{Af}{2}^2 - \ebb \norm{Af}{2}^2 \Big| \\ 
		& \geq \frac{c_\psi}{C_\psi} 
		- \sup_{f \in \sqrt{s'} B_{\ell^1}^N \cap S^{N-1}} 
		\Big| \norm{Af}{2}^2 - \ebb \norm{Af}{2}^2 \Big| \; .
	\end{align*}
	By Theorem~\ref{prop:main} applied for $\sqrt{s'} B_{\ell^1}^N \cap 
	S^{N-1} \subseteq \sqrt{s'} B_{\ell^1}^N$ and 
	$\delta = c_\psi/(64 \cdot c_0 \,\max\{1,C_\psi\})$,
	where $c_0>0$ is the constant from Theorem~\ref{prop:main},
	the supremum is bounded with probability at least $1- 2
	\exp(-c_1(c_\psi/\max\{1,C_\psi\})^2 m / (s K_\psi^2) )$
	by 
	\begin{equation}
	\label{eq:proof_riesz_good}
	\begin{split}
		\sup_{f \in \sqrt{s'} B_{\ell^1}^N \cap S^{N-1}}& 
	 	\Big| \norm{Af}{2}^2 - \ebb \norm{Af}{2}^2 \Big|  \\
		&\leq \frac{16 \cdot c_0 \, c_\psi}{64 \cdot c_0 \,C_\psi \max\{1,C_\psi\}} 
	 	\Big(1 + \sup_{f \in \sqrt{s'} B_{\ell^1}^N \cap S^{N-1} } 
		 \ebb |\inner{f}{X}|^2 \Big) \\
		 &\leq \frac{c_\psi}{4C_\psi \max\{1,C_\psi\}} (1 + C_\psi)
	\end{split}
	\end{equation}
	provided that (recall that $s' = 16s$) for a constant
	$c>0$ depending on $c_0$
	\begin{displaymath}
		m \geq c \cdot \left(\frac{\max\{1,C_\psi\}}{c_\psi} \right)^2
		K_\psi^2\; s \, \log(s  K_\psi^2 \max\{1,C_\psi\}/c_\psi)^2 \log(eN) \; .
	\end{displaymath}
	We distinguish two cases on the event \eqref{eq:proof_riesz_good}. For 
	$C_\psi\leq 1$, we find
	\begin{equation}
		\label{eq:proof_riesz_g1}
		\frac{c_\psi}{4C_\psi \max\{1,C_\psi\}} (1 + C_\psi)
		 \leq \frac{c_\psi}{4 C_\psi} (1 + C_\psi)
		 \leq \frac{c_\psi}{2 C_\psi} \; .
	\end{equation}
	On the other hand, for $C_\psi>1$,
	\begin{equation}
		\label{eq:proof_riesz_l1}
		\frac{c_\psi}{4C_\psi \max\{1,C_\psi\}} (1 + C_\psi)
		 \leq \frac{c_\psi}{4 C_\psi^2} (1 + C_\psi)
		 \leq \frac{c_\psi}{2 C_\psi} \; .
	\end{equation}
	Combining the estimates \eqref{eq:proof_riesz_good}, 
	\eqref{eq:proof_riesz_g1} and \eqref{eq:proof_riesz_l1} concludes 
	the proof.
\end{proof}

\subsection{Coherence-based sampling}
\label{sec:finite_Riesz}
We conclude this section with an application of Theorem~\ref{prop:main}
that will be essential for our analysis of the \textsf{CORSING} method
in Section~\ref{sec:PDEs}. 
We consider vectors $\{b_1,\ldots, b_M\} \subseteq \cbb^N$ with $M \geq N$, 
such that the matrix
\begin{equation}
  \label{eqn:finite_riesz_matrix}
  B =
  \begin{pmatrix}
    b_1 \;|& b_2 \;|& \cdots \; |&  b_M
  \end{pmatrix}^T \in \mathbb{C}^{M \times N},
\end{equation}
satisfies 
\begin{equation}
\label{eqn:finite_riesz}
 c_B \|f\|_2^2 \leq \|B f\|_2^2 \leq C_B 
 \|f\|_2^2, \qquad \text{for all} \quad f \in \Sigma_{s,N},
\end{equation}
and suitable constants $0 < c_B \leq C_B$.
The constants $c_B$ and $C_B$ 
are also known as the minimum and the maximum $s$-sparse eigenvalues of $B^*B$ 
(see \cite{rudelson2012reconstruction}). Our goal is to construct a matrix that satisfies the restricted isometry property
by sampling a small number of rows of $B$. Following \cite{krahmer2014stable},
the idea is to sample according to the local coherence. 
We say that $\nu\in\mathbb{R}^M$ is a local coherence for $B$ if 
\begin{equation}
\label{eq:def_nu}
\max_{n \in [N]} |B_{jn}|^2 \leq \nu_j \; \qquad \text{for all } j \in [M] \; .
\end{equation}
Define a random vector $X \in \cbb^N$ by 
\begin{equation}
	\pbb \Big( 
		X = \sqrt{\frac{\norm{\nu}{1}}{\nu_j}} \; b_j 
	\Big)
	= \frac{\nu_j}{\norm{\nu}{1}} \; .
\end{equation}
Let $X_1,\ldots, X_m$ denote independent copies of $X$ and consider
\begin{equation}
\label{eq:def_A_finite_Riesz}
A = \frac{1}{\sqrt{m C_B}}\begin{pmatrix}
X_1 
|&\cdots &|
X_m
\end{pmatrix}^T.
\end{equation}
Then for any $f \in \cbb^N$, 
\begin{equation*}
  \ebb |\scalar{X}{f}|^2 
  =  \sum_{j=1}^M \frac{\nu_j}{\norm{\nu}{1}}
  \frac{\norm{\nu}{1}}{\nu_j}|\scalar{b_j}{f}|^2
  =  \norm{Bf}{2}^2 \; 
  \quad \text{and} \quad
  \frac{1}{m} \sum_{i = 1}^m |\scalar{X_i}{f}|^2 = C_B \|Af\|_2^2 \; .
\end{equation*}
Moreover, for each $i \in [N]$ and $e_i$ a standard basis vector
of $\cbb^m$ we have
\begin{equation}
	|\inner{X}{e_i}| \leq 
  \sqrt{\norm{v}{1}}\max_{j \in[M]} \frac{1}{\sqrt{\nu_j}} \max_{n \in [N]}
  |B_{jn}| \leq \sqrt{\norm{\nu}{1}} \; ,
\end{equation}
Now, applying Theorem~\ref{prop:main} and arguing similarly to 
Theorem~\ref{thm:riesz_rip}, we obtain the following RIP result for the matrix 
$A$.

\begin{thm} 
\label{thm:CORSING_rip}
There exist universal constants $c_0,c_1 >0$, such that the following holds.
Let $B \in \mathbb{C}^{M \times N}$ be such that \eqref{eqn:finite_riesz} 
holds.  Let $s,m \in \mathbb{N}$ and $\ep \in (1-\frac{c_B}{C_B},1)$. Let $\nu 
\in \mathbb{R}^M$ be such that \eqref{eq:def_nu} holds and assume 
$$
m \geq 
c_0 \, \max\{C_B^{-2},1\} 
\, \eta^{-2}  \, 
s \norm{\nu}{1} \log^2(s \norm{\nu}{1} \max\{C_B^{-2},1\} \eta^{-2}) \log(eN) \; ,
$$
with $\eta = \ep - 1 +\frac{c_B}{C_B}$. 
Then, with probability at least $
1-2\exp(- c_1 \min\{C_B^2,1\} \eta^2  m / (
\norm{\nu}{1} s))
$ the matrix $A$ defined in \eqref{eq:def_A_finite_Riesz} satisfies RIP$(\ep,s)$. 
\end{thm}

\subsection{Extension to weighted $\ell^1$-minimization}
\label{sec:weighted_l1}

The previous sections studied the impact of Theorem~\ref{prop:main}
on sparse recovery via $\ell^1$-minimization. The weighted 
$\ell^1$-minimization program (see \eqref{def/program/weighted} below)
was suggested in \cite{Rauhut2015} as a means of 
incorporating additional information on the smoothness of the function 
$F: S \to \cbb$, when trying to solve the interpolation problem discussed 
in Section~\ref{sec:Riesz}. Recall that we aim at finding the coefficients
of a function $F: S \to \cbb$ with respect to a bounded Riesz system 
$(\psi_j)_{j \in \mathbb{N}}$ from a finite number noisy observations $F
(\omega_1) + e_1 ,\dots, F(\omega_m) + e_m$. Our strategy is based on the
assumption that the sequence of coefficients of $(f_j)_{i \in [N]}$, which
represent the function $F = \sum_{j \in [N]} f_j \psi_j$ with respect to
$(\psi_j)_{j \in [N]}$, is $s$-sparse. In real world scenarios smoothness
of $F$ often plays a crucial role. In this section we discuss
the problem of amalgamating smoothness assumptions with sparsity
assumptions on $F$. Weighted sparsity and weighted $\ell_1$-minimization is also crucial for compressive sensing approaches to function recovery in high-dimensions and, in particular, for solving parametric PDEs \cite{rauhut2017compressive,webster,borasc17} appearing for instance in the context of uncertainty quantification.
We refer the reader to \cite{Rauhut2015,rauhut2017compressive,webster,adcock2017compressed} and the references therein
for an introduction to the subject and examples. In this section our path
to recovery guarantees for weighted $\ell^1$-minimization passes through
a suitable weighted notion of the restricted isometry property as used
in \cite{Rauhut2015}. 

\subsubsection{Weighted $\ell^p$-spaces and the weighted RIP}
In order to quantify smoothness information on the function that we are trying
to interpolate, we follow the reference \cite{Rauhut2015} and use the following 
weighted versions of $\ell^p$-spaces for $0<p \leq 2$. 
Let $w = (w_j)_{j \in [N]}$ denote a sequence of weights with $w_j \geq 1$ 
and for $0< p \leq 2$ set
\begin{equation}
    \ell_w^p := \Big\{ f \in \cbb^N \; : \; 
    \norm{f}{w,p} := \Big( \sum_{j \in [N]} w_j^{2-p} |f_j|^p \Big)^{\frac{1}{p}} \Big \}
    \; .
    \label{eq:def/weighted}
\end{equation}
Later the weight sequence $(w_j)_{j \in [N]}$ is chosen in a way, such that
$w_j \geq \norm{\psi_j}{L^\infty(S)}$. The work \cite{Rauhut2015} recognized that
there is a notion of sparsity which is consistent with the weighted version of
$\ell^p$-spaces in \eqref{eq:def/weighted}. For each $f \in \cbb^N$ we define a 
suitiable version of the $\ell^0$-norm given by setting 
\begin{equation}
	\norm{f}{w,0} = \sum_{j \in \{j : f_j \neq 0\} } w_j^2 \;, 
\end{equation}
and calling a signal $f \in \cbb^N$ $s$-sparse repect to the sequence $w$, if 
$\norm{f}{w,0} \leq s$. 

The notion of sparsity with respect to a weight sequence $w \in [1,\infty)^N$
gives rise to the following version of the restricted isometry property.
\begin{defn}[Weighted restricted isometry property]
	Let 
	\begin{equation}
	    D_{s,N}(w) := \{ f \in \cbb^N \; : \; \norm{f}{w,0} \leq s, \, \norm{f}{2} = 1\} \; .
	\end{equation}
	The weighted restricted isometry constant $\ep_{s,w}$ of a matrix $A \in \cbb^{m \times N}$ is defined by
	\begin{equation}
		\ep_{s,w} := \sup_{f \in D_{s,N}(w)} \Big| \|Af\|_2^2 - 1\Big|.
		\label{eq:definition_w-rip}
	\end{equation}
	If a matrix $A \in \cbb^{m \times N}$ satisfies 
	$\ep_{s,w} \leq \ep$ for some $\ep>0$, then 
	we say that $A$ satisfies $w$-RIP$(\ep,s)$. 
\end{defn}
A crucial observation is that the Cauchy-Schwarz inequality implies that every $f \in D_{s,N}(w)$ satisfies
\begin{equation*}
    \norm{f}{w,1} \leq \Big( \sum_{j=1}^N w_j^2 \Big)^{1/2}
    \norm{f}{2} \leq \sqrt{s}  \; .
\end{equation*}
Hence, we find the familiar inclusion
\begin{equation}
	D_{s,N}(w) \subset \sqrt{s} B_{\ell_w^1}^N \cap S^{N-1} \; .
	\label{eq:inclusion_weighted_sparsity}
\end{equation}
\subsubsection{Recovery by weighted $\ell^1$-minimization}
In Section~\ref{sec:Riesz} we leveraged a connection between the restricted isometry
property and recovery guarantees of $\ell^1$-minimization in order to use 
Theorem~\ref{prop:main}. A similar connection was observed for the weighted
$\ell^1$-minimization program,
\begin{equation}
    \label{def/program/weighted}
	\min_{z \in \cbb^N} \norm{z}{w,1} \qquad \text{subject to} \qquad 
	\norm{ \sqrt{C_\psi m} Az - y}{2} \leq \zeta \; .
	\tag{BP$_{w,\zeta}$}
\end{equation}
Let us state the following simplified version of this connection from
\cite{Rauhut2015}.
\begin{thm}[{\cite[Theorem 4.5 \& Corollary 4.3]{Rauhut2015}}]
	Let $A \in \cbb^{m \times N}$ such that $A$ satisfies $w$-RIP$(\ep,s)$
	for $\ep \leq 1/3$ and for $s \geq  2\norm{w}{\infty}^2$. 
	Let $f \in D_{s,N}(w)$ and set $y = A f + e$ with $\norm{e}{2} \leq \zeta$.
	Further, let $f^\#$ denote the minimizer of \eqref{def/program/weighted}, then
	there exists constants $c_0,c_1>0$ depending only on $\ep>0$, such that
	\begin{align*}
		\norm{f - f^\#}{w,1} &\leq c_0 \sqrt{s} \zeta \, ,\\
		\norm{f - f^\#}{2} &\leq c_1 \zeta \, .
	\end{align*}
\end{thm} 
Using this result we can employ the same strategy as in Section~\ref{sec:Riesz}
in order to show recovery guarantees for the weighted $\ell^1$-minimization
program, provided that we have a suitable version of Theorem~\ref{prop:main}.

\subsubsection{Weighted restricted isometry property for Riesz frames}
Recalling the inclusion \eqref{eq:inclusion_weighted_sparsity}, we will use the
following result in order to obtain recovery guarantees for the program 
\eqref{def/program/weighted}.
\begin{thm}
\label{thm:weighted_main}
There exists absolute constants $\kappa, c_0, c_1 >0$ such that the following 
holds.
Let $X_1,\ldots,X_m$ be independent copies of a random vector $X \in \cbb^N$ 
such that for all $j = 1,\ldots,N$ we have $|\inner{X}{e_j}|\leq K_j$ for 
some $K_j>0$ where $e_1,\ldots,e_N$ is the standard basis of $\mathbb{C}^N$.  
Let $w_j \geq K_j$ for all $j = 1, \ldots, N$ and let $T\subseteq \{ f \in 
\cbb^N :\norm{f}{w,1} \leq \sqrt{s} \}$, $\delta \in (0,\kappa)$. 
Assume further that
\begin{equation}
\label{eqn:numSampMain/w_l1}
 m \geq c_0 \, \delta^{-2} s \log(e N) \log^2(s/\delta) \; .
\end{equation}
Then, with probability exceeding $1-2\exp(-\delta^{2} m /s)$,
\begin{equation}
  \label{eq:main_result/w_l1}
 \sup_{f \in T} \Big| \frac{1}{m}\sum_{i=1}^m 
 |\inner{f}{X_i}|^2 - \ebb |\inner{f}{X}|^2 \Big| \leq c_1 
 \Big(\delta + \delta \sup_{f \in T} \ebb|\inner{f}{X}|^2 \Big)\; .
\end{equation}
\end{thm}
This enables us to show the following result on the $w$-RIP$(\ep,s)$ for
truncated Riesz sequences (and in particular for bounded orthonormal systems).
\begin{thm}
\label{prop:weighted_riesz}
There are absolute constants $c_0 , c_1 > 0$, such that the following holds.
Let $H = L^2 (S, \mu)$ and let $(\psi_j )_{j\in [N]}$ be a truncated Riesz system. 
Let $\ep \in (1-\frac{c_\psi}{C_\psi},1)$, $w_j \geq \norm{\psi_j}{L^\infty(S)}$ and 
assume that
\begin{equation}
	m \geq c_0 \, \max\{C_{\psi}^{-2}, 1\} \eta^{-2} s 
	\log(s \max\{C_\psi^{-2}, 1\} \eta^{-2} )^2 \log(eN )\; ,
\end{equation}
where $\eta = (\ep-1+ C_\psi ) > 0$. Then, with probability at least 
$1-2 \exp(-c_1 \min\{C_\psi , 1\}\eta^2 m/ s)$ the matrix $A$ defined in 
\eqref{eq:riez/measurement} satisfies $w$-RIP$(\ep, s)$.
\end{thm}
Theorem~\ref{prop:weighted_riesz} follows from Theorem~\ref{thm:weighted_main}
by chasing through to argument of Theorem~\ref{thm:riesz_rip}. and replacing
the application of Theorem~\ref{prop:main} to $D_{s,N} \subseteq \sqrt{s} 
B_{\ell^1}^N$ by applying Theorem~\ref{thm:weighted_main} to the weighted
version $D_{s,N}(w) \subseteq \sqrt{s} B_{\ell_w^1}^N$.

We will outline how to deduce Theorem~\ref{thm:weighted_main} after we have
established the proof of Theorem~\ref{prop:main}. 
The difference of the proofs of Theorem~\ref{prop:main} and of Theorem~\ref{thm:weighted_main} is limited
to technical details.

\section{Application to numerical approximation of PDEs}

\label{sec:PDEs}


As an application of the theory presented in Section~\ref{sec:results}, we 
consider the \textsf{CORSING} (\textsf{COmpRessed SolvING}) method for the 
numerical approximation of solutions of PDEs
based on compressive sensing \cite{simone_thesis,Brugiapaglia2015,brugiapaglia2018theoretical}. The \textsf{CORSING} method is a general paradigm to compute 
a sparse approximation to the solution of a PDE that admits a weak formulation.
It assembles a reduced Petrov-Galerkin discretization via compressive 
sensing. As discussed in \cite{brugiapaglia2018theoretical}, the advantages of 
this method compared to other nonlinear approximation methods for PDEs, such as 
adaptive finite elements or adaptive wavelets methods (see, e.g., 
\cite{urban2009wavelet} and references therein), are that (i) no \emph{a posteriori} 
error indicators are needed and (ii) the assembly of the discretization matrix 
and the sparse recovery, here performed via Orthogonal Matching Pursuit (OMP), 
can be easily parallelized. 

Here, we will focus on the restricted isometry analysis and sparse recovery 
guarantees for \textsf{CORSING}. For numerical experiments for multi-dimensional 
advection-diffusion-reaction equations and for the Stokes problem, we refer the 
reader to \cite{Brugiapaglia2015,simone_thesis,brugiapaglia2018theoretical,
Brugiapaglia2018wavelet}. We also note in passing that the \textsf{CORSING} 
paradigm can be adapted to the framework of collocation techniques for PDEs 
(see \cite{brugiapaglia2018compressive}). For an overview of numerical methods 
for PDEs based on compressive sensing, we refer the reader to \cite[Section 1.2]{brugiapaglia2018compressive}.

Concerning the sparse recovery method, we resort to OMP as opposed to 
$\ell^1$-minimization for two main reasons, mainly related to computational efficiency 
considerations: (i) given a target sparsity level, using OMP we can easily control 
the number of iterations and, consequently, the computational cost of the recovery 
phase; (ii) OMP is easily parallelizable. For a numerical comparison between OMP and 
$\ell^1$-minimization in this context, we refer to \cite[Section 5]{Brugiapaglia2015}. 
Finally, we note that the restricted isometry analysis of \textsf{CORSING} presented 
here can be applied when different sparse recovery procedures are considered.

The section is structured as follows. After describing the setting of weak 
problems in Hilbert spaces in Section~\ref{sec:weak}, we recall the main 
elements of the \textsf{CORSING} method in Section~\ref{sec:CORSING}. Then, 
we present an RIP analysis for the \textsf{CORSING} discretization matrix 
in Section~\ref{sec:CORSING_RIP} and discuss recovery guarantees for the method 
when the solution is approximated via OMP in 
Section~\ref{sec:CORSING_OMP}. 

\subsection{Weak problems in Hilbert spaces}
\label{sec:weak}
Let $U$ and $V$ be separable Hilbert spaces equipped with inner products 
$(\cdot,\cdot)_U$ and $(\cdot,\cdot)_V$, and norms $\|\cdot\|_U= 
(\cdot,\cdot)_U^{1/2}$ and $\|\cdot\|_V = (\cdot,\cdot)_V^{1/2}$. We consider a 
weak problem of the form
\begin{equation}
\label{eq:weakprob}
\text{find } u \in U : \quad a(u,v) = \mathcal{F}(v), \quad \forall v \in V,
\end{equation}
where $a : U \times V \to \RR$ is a bilinear form and $\Fop \in V^*$,  
the dual space of $V$. An important example of \eqref{eq:weakprob} is the 
weak formulation of the advection-diffusion-reaction equation with homogeneous 
boundary conditions, where $U = V =H_0^1(\mathcal{D})$, $\mathcal{D} 
\subseteq \mathbb{R}^d$ is the physical domain, and the bilinear form is
\begin{equation}
\label{eq:def_a_ADR}
a(u,v) = \int_{\mathcal{D}} \mu \nabla u \cdot \nabla v + (\beta \cdot \nabla u) v + 
\rho uv \de x, \quad \forall u,v \in H_0^1(\mathcal{D}),
\end{equation}
where $\mu : \mathcal{D} \to \mathbb{R}$, $\beta : \mathcal{D} \to 
\mathbb{R}^d$, and $\rho:\mathcal{D} \to \mathbb{R}$ are the diffusion, 
advection, and reaction coefficients. The operator $\Fop$ is defined by 
$$
\Fop(u) := 
\int_\mathcal{D} u F \de x,\quad \forall u \in H_0^1(\mathcal{D}),
$$ 
where $F : \mathcal{D} \to \mathbb{R}$ is the forcing term (see, e.g., 
\cite{QuarteroniValli} for more details). The analysis in this section will 
focus on abstract weak problems of the form \eqref{eq:weakprob}, but it can be 
specified to particular PDEs, such as  
\eqref{eq:def_a_ADR} (see \cite{brugiapaglia2018theoretical}).

In order to guarantee the existence and uniqueness of the solution to 
\eqref{eq:weakprob}, we assume that the bilinear form $a(\cdot,\cdot)$ satisfies the 
hypotheses of the classical Babu\v{s}ka-Ne\v{c}as theory (see, e.g.,  
\cite[Theorem 5.1.2]{QuarteroniValli}), namely
\begin{align}
\label{eq:infsup}
& \exists \alpha > 0 :\quad \inf_{u \in U\setminus \{0\}} \sup_{v \in V\setminus 
\{0\}} \frac{a(u,v)}{\|u\|_U \|v\|_V} \geq \alpha,\\
\label{eq:continuity}
&\exists \beta >0 :\quad  \sup_{u \in U\setminus \{0\}} \sup_{v \in V\setminus 
\{0\}} \frac{a(u,v)}{\|u\|_U \|v\|_V} \leq \beta, \\
\label{eq:surjectivity}
&\sup_{u \in U} a(u,v) > 0 , \quad \forall v \in V \setminus \{0\}.
\end{align}

\newcommand{\trial}{\varphi}
\newcommand{\Trial}{\Phi}
\newcommand{\test}{\xi}
\newcommand{\Test}{\Xi}

Consider two Riesz bases $(\trial_j)_{j \in \NN}$ and  $(\test_q)_{q 
\in \NN}$ for $U$ and $V$, respectively, i.e., satisfying relation \eqref{def/riesz}. 
We denote the lower and the upper Riesz 
constants of $(\trial_j)_{j \in \NN}$ as  $c_\trial$ and $C_\trial$, 
respectively. Analogously, we denote the Riesz constants of 
the system $(\test_q)_{q \in \NN}$ as $c_\test$ and  $C_\test$.
Finally, define the reconstruction and decomposition operators 
$\Trial:\ell^2(\NN) \to U$ and $\Trial^*: U \to \ell^2(\NN)$, respectively, by 
\begin{equation*}
\Trial x = \sum_{j \in\NN} x_j \trial_j, \; \forall x \in \ell^2(\NN),
 \quad (\Trial^* u)_j = (u,\trial_j^*)_U, \; \forall j \in \NN,
\end{equation*}
where $(\trial_j^*)_{j \in \NN}$ is the biorthogonal basis of $(\trial_j)_{j 
\in \NN}$ (see \cite{Christensen2002}). The operators $\Test$ and $\Test^*$ are 
defined analogously. By the Riesz property,
\begin{equation}
	c_\varphi \norm{x}{2}^2 \leq  \norm{\Phi x}{2}^2 \leq C_\varphi \norm{x}{2}^2,
	\quad
	c_\xi \norm{x}{2}^2 \leq \norm{\Xi x}{2}^2 \leq C_\xi \norm{x}{2}^2 \; .
	\label{eq:riesz_decomposition}
\end{equation}
We discretize problem \eqref{eq:weakprob} via a Petrov-Galerkin approach. Let 
$N, M \in \NN$ and consider the finite dimensional truncated spaces 
\begin{equation}
	U^N := \Span(\trial_j)_{j \in[N]}, \quad V^M := \Span(\test_q)_{q \in[M]},
\end{equation}
called the \textit{trial} and the \textit{test space}, respectively. Accordingly
we call $(\trial_j)_{k \in [N]}$ and $(\test_j)_{j \in [N]}$ trial and
test basis functions, respectively. We associate with these spaces a finite dimensional 
formulation of \eqref{eq:weakprob}, namely
\begin{equation}
\label{eq:PGformul}
\text{find } u \in U^N : \quad a(u,v) = \mathcal{F}(v), \quad \forall v \in V^M.
\end{equation}
By the bilinearity of $a(\cdot,\cdot)$ and the linearity of $\F$,
\eqref{eq:PGformul} can be discretized as a linear system
\begin{equation}
\label{eq:PGsystem}
B z = c,
\end{equation}
where $B \in \mathbb{C}^{M \times N}$ and  $c \in \mathbb{C}^{M}$ are defined by 
\begin{equation}
B_{qj} = a(\trial_j, \test_q),\quad  c_q = \Fop(\test_q), \quad \forall j \in 
[N], \; \forall q \in [M].
\end{equation} 
The linear system \eqref{eq:PGsystem} is usually referred to as a 
Petrov-Galerkin discretization of \eqref{eq:weakprob}, or, in particular, as a 
Galerkin discretization when $N=M$ and $U^N = V^M$. This general class of 
discretizations contains many popular numerical approximation methods for PDEs, 
whose most prominent example is the finite element method (see, e.g., \cite{QuarteroniValli}).

\begin{rem}[On the Riesz basis assumption]
When the physical domain $\mathcal{D}$ has dimension $d = 1$, it is easy to construct 
orthonormal bases for the trial and the test functions. For example, in the case 
of homogeneous boundary conditions, i.e., $U = V = H_0^1(\mathcal{D})$ with $\mathcal{D} 
= (0,1)$, we can  consider hierarchical hat functions (see, e.g., \cite{Dahmen2003}) 
as trial basis functions and Fourier-like functions (e.g., sine functions) as test 
basis functions. This setting, considered in 
\cite{Brugiapaglia2015,brugiapaglia2018theoretical}, leads to two orthonormal bases 
of $H_0^1(\mathcal{D})$. When $\mathcal{D} = (0,1)^d$ with $d > 1$, generalizing 
the Fourier-like basis while preserving orthogonality is easily done via tensorization. 
However, tensorizing the hierarchical basis of hat functions does not preserve 
orthogonality. In order to obtain a stable discretization, one can instead resort 
to biorthogonal spline wavelets to obtain a Riesz basis (see \cite{Brugiapaglia2018wavelet} 
for more details).
\end{rem}

\subsection{The \textsf{CORSING} method}
\label{sec:CORSING}

In order to take advantage of the compressive sensing paradigm, we consider a 
Petrov-Galerkin discretization of \eqref{eq:weakprob}, where the bilinear form 
is evaluated at trial and test basis functions that satisfy suitable incoherence 
properties. Solving the discritized linear system in \eqref{eq:PGsystem} can be
expensive from a computational perspective. 
%
%
The idea of \textsf{CORSING} is to 
solve the discretized system \eqref{eq:PGsystem} via compressive sensing, 
with the aim of computing an $s$-sparse approximation to the solution to 
\eqref{eq:weakprob}, with $s \ll N$, i.e., the computed solution 
belongs to the space
\begin{equation*}
U_s^N := \bigg\{\sum_{j \in [N]} w_j \trial_j : \|w\|_0 \leq s\bigg\}.
\end{equation*}
The quality of the approximation space $U^N_s$ depends on the sparsity or compressibility of 
the solution with respect to the trial basis. When the trial basis is a Fourier-like 
basis, sparsity or compressibility are observed when 
the most important frequencies of the solution are clustered over different 
regions of the spectrum due to multiscale phenomena  (see, e.g.,  
\cite{mackey2014compressive}). On the other hand, using a hierarchical (or wavelet) 
basis as a trial basis leads to sparsity or compressibility of solutions with 
localized features or boundary layers (see \cite{simone_thesis,brugiapaglia2018theoretical}).

In order to reduce the dimensionality of the  discretization \eqref{eq:PGsystem}, 
we pick $m \ll M$ test indices $\tau_1 \ldots, \tau_m\in[M]$ 
i.i.d. at random according to a discrete probability measure $p \in \mathbb{R}^M$ on the 
index set $[M]$, i.e.,  
\begin{equation}
\label{eq:defprob}
\Prob(\tau_i = q)=p_q, \qquad  \text{for } i \in[m], \; q \in [M].
\end{equation} 
Then, the resulting \textsf{CORSING} discretization of \eqref{eq:PGformul} is given by the underdetermined linear system
\begin{equation}
\label{eq:CORSINGsys}
A z = y,
\end{equation}
where $A \in \mathbb{C}^{m \times N}$ and $y \in \mathbb{C}^m$ are defined by 
\begin{equation}
\label{eq:defAf}
A_{ij} = a(\trial_j, \test_{\tau_i}), \quad y_i = \Fop(\test_{\tau_i}), 
\quad\forall i \in[m],\; \forall j \in[N].
\end{equation} 
Taking advantage of the compressive sensing paradigm, in the recovery phase we 
seek an $s$-sparse solution to \eqref{eq:CORSINGsys}, by considering an approximate 
solution $\hat{x} \in \mathbb{C}^N$ to the optimization program
\begin{equation}
\label{eq:P0s}
\hat{x} \approx \arg\min_{z \in \Sigma_{s,N}} \|D(A z - y)\|_2,
\end{equation}
where $D \in \RR^{m \times m}$ is a diagonal preconditioning matrix with diagonal
elements 
\begin{equation}
\label{eq:defD}
D_{ii} = \frac{1}{\sqrt{mp_{\tau_i}}}, \qquad \text{for } i \in[m].
\end{equation}
The \textsf{CORSING} solution is then defined by 
\begin{equation}
\label{eq:defuhat}
\hat{u} := \Trial \hat{x} = \sum_{j \in [N]} \hat{x}_j \trial_j.
\end{equation}
A detailed discussion of the case where $\hat{x}$ is computed  via OMP will be carried out in Section~\ref{sec:CORSING_OMP}.

\subsection{Restricted isometry property}
\label{sec:CORSING_RIP}
A first theoretical analysis of the \textsf{CORSING} method was proposed in 
\cite{brugiapaglia2018theoretical}. The main tool employed is the local 
$a$-coherence, a generalization of the local coherence \cite{Adcock2013,krahmer2014stable} to bilinear forms in 
Hilbert spaces that implements a preconditioning 
as suggested in \cite{rauhut2012sparse}.

\begin{defn}[local $a$-coherence] 
\label{def:localacohe}
Given $N \in \mathbb{N} \cup \{\infty\}$, the sequence $\mu^N := (\mu_q^N)_{q 
\in \NN}$ defined by 
\begin{equation}
\label{eq:deflocala-cohe} 
\mu_q^N := \sup_{j \in [N]} |a(\trial_j, \test_q)|^2, \quad q \in \NN,
\end{equation}
is called the \textit{local $a$-coherence} of $(\trial_j)_{j \in [N]}$ 
with respect to $(\test_q)_{q \in \mathbb{N}}$.
\end{defn}
In the following, we will assume that $\mu^N \in \ell^1(\NN)$, for every $N \in 
\NN$. 
The next proposition gives a sufficient condition on $\mu^N$ and $M$ 
in order for $U^N_s$ and $V^M$ to satisfy the inf-sup 
condition. The result immediately follows from 
\cite[Lemma 3.6]{brugiapaglia2018theoretical} and 
\cite[Remark 3.11]{brugiapaglia2018theoretical}.

\begin{prop}
\label{prop:truncation}
Let $s,N,M\in\NN$, with $s \leq N$, and $\deltaM \in (0,1)$ be such that 
\begin{equation}
\label{eq:trunccond}
\sum_{q > M}  \mu_q^N 
\leq \frac{\alpha^2 \deltaM c_\trial c_\test}{s}.
\end{equation}
Then, 
\begin{equation}
\inf_{u \in U^N_s} \sup_{v \in V^M} \frac{a(u,v)}{\|u\|_U \|v\|_V} \geq 
(1-\deltaM)^{\frac12}\alpha.
\end{equation}
\end{prop}

In the following we will assume to have access to a computable (but not necessarily 
sharp) upper bound $\nu^N = (\nu^N_q)_{q \in \NN}$ to the local $a$-coherence,
i.e., $\mu_q^N \leq \nu_q^N$ for all $q \in \NN$ and define 
\begin{equation}
\label{eq:def_nu_MN}
\nu^{N,M} := 
(\nu^N_q)_{q \in [M]}\in\RR^M \; .
\end{equation} 
Throughout this section, we choose the probability density $p$ in \eqref{eq:defprob} over the test indices 
as 
\begin{equation}
p := \frac{\nu^{N,M}}{\|\nu^{N,M}\|_1}.
\end{equation}

The analysis carried out in \cite{brugiapaglia2018theoretical} provides 
sufficient conditions on $M$ and $m$ that guarantee an optimal error estimate 
in expectation for \textsf{CORSING}. Yet, the recovery analysis in 
\cite{brugiapaglia2018theoretical} has two main 
limitations: (i) $m$ depends quadratically on $s$ (up to logarithmic factors), 
whereas from the compressive sensing theory one expects this dependence to 
be linear (up to logarithmic factors); (ii) the vector $\hat{x}$ is assumed to 
solve \eqref{eq:P0s} \emph{exactly}, whereas in practice OMP (or another approach 
for sparse recovery) has to be employed to approximate the solution to \eqref{eq:P0s}.

These two issues are fixed in the following. In particular, thanks to the 
analysis based on the restricted isometry property, we show that a linear 
dependence of $m$ on $s$ (up to logarithmic factors) is sufficient and prove a 
recovery error estimate for the case that the \textsf{CORSING} solution 
$\hat{x}$ is approximated via OMP (Section~\ref{sec:CORSING_OMP}).

We start by defining a condition number of the weak infinite-dimensional problem 
\eqref{eq:weakprob} by
\begin{equation}
\label{eq:defzeta}
\condweak := \frac{C_\trial C_\test \beta^2}{c_\trial c_\test\alpha^2},
\end{equation}
where $\alpha$ and $\beta$ are the inf-sup and continuity constants of $a(\cdot,\cdot)$ 
defined in \eqref{eq:infsup} and \eqref{eq:continuity}, respectively. Of course, in
general we have $\condweak \geq 1$. 

We also introduce the preconditioned and rescaled versions of  $A$ and $y$ 
defined in \eqref{eq:defAf} by
\begin{equation}
\label{eq:Aftilde}
\tilde{A}:=\frac{1}{\beta \sqrt{C_\trial C_\test} } DA, \quad 
\tilde{y}:=\frac{1}{\beta \sqrt{C_\trial C_\test}}D y,
\end{equation} 
where $D$ is defined as in \eqref{eq:defD}. This normalization ensures that 
$\Expe\|\tilde{A} t\|_2^2 \leq \|t\|_2^2$ for every $t \in \sparse$. In this 
setting, the following RIP result holds. 

\begin{thm}[RIP for the \textsf{CORSING} matrix]
\label{thm:RIPcorsing}
There exist universal constants $c_0,c_1>0$ such that the following holds. 
Let $s,N \in \NN$ with $s \leq N$,  $\deltaM\in(0,1)$ and choose $M= M(s,N, \mu^N)$
such that 
$$
\sum_{q > M} \mu^N_q \leq \frac{\alpha^2 \gamma c_\trial c_\test}{s}.
$$
Further, let $\condweak$ be defined as in \eqref{eq:defzeta} and let $\nu^{N,M}$ is defined as in \eqref{eq:def_nu_MN}.
Then, for every 
\begin{equation}
\label{eq:RIPthm:defdeltam}
\ep \in \bigg(1- \frac{1-\deltaM}{\condweak},1\bigg),
\quad \text{and accordingly} \quad
\eta = \ep -1 +\frac{1-\gamma}{\condweak}
\end{equation} 
with probability at least
$$
1-2 \exp\left(-c_1 \frac{\eta^2 m \min\{1,C_\varphi^2C_\xi^2\beta^{4} \}}{
s \|\nu^{N,M}\|_1^2}
\right)
$$
the rescaled \textsf{CORSING} matrix $\tilde{A}$ 
defined in \eqref{eq:Aftilde} satisfies the $\RIP{\varepsilon}{s}$ provided that
\begin{equation}
\label{eq:CORSING_sampling_complexity}
m \geq c_0
\frac{
\|\nu^{N,M}\|_1^2}{\min\{1,C_\varphi^{2}C_\xi^{2} \beta^{4} \}\,\eta^2}  s \log(e N) 
\log^2\left(\frac{s\|\nu^{N,M}\|_1^2}{\min\{1,C_\varphi^{2}C_\xi^{2} \beta^{4} \} \, \eta^2}\right) \; .
\end{equation}
\end{thm}
\begin{proof}
This theorem is a consequence of Theorem~\ref{thm:CORSING_rip}. In order to 
apply this result, we need to find constants $c_B$ and $C_B$ such that  
\eqref{eqn:finite_riesz} holds.

Let us start by estimating $c_B$. Thanks to Proposition~\ref{prop:truncation} and 
to the Riesz property \eqref{eq:riesz_decomposition} of $(\trial_j)_{j \in \mathbb{N}}$ 
and $(\test_q)_{q \in \NN}$, we have
\begin{align*}
\inf_{x \in \sparse\setminus\{0\}} \frac{\|Bx\|_2^2}{\|x\|_2^2}
 & \geq c_{\trial} \inf_{x \in \sparse\setminus\{0\}} \frac{\|B 
x\|_2^2}{\|\Trial x\|_U^2}
 = c_{\trial} \bigg(\inf_{x \in \sparse\setminus\{0\}} \sup_{z \in 
\mathbb{R}^{M}\setminus\{0\}}\frac{z^T B x}{\|\Trial x\|_U\|z\|_2} \bigg)^2\\
&\geq c_{\trial} c_\test \bigg(\inf_{x \in \sparse\setminus\{0\}} \sup_{z \in 
\mathbb{R}^{M}\setminus\{0\}}\frac{z^T B x}{\|\Trial x\|_U\|\Test z\|_2} 
\bigg)^2\\
& = c_\trial c_\test \bigg(\inf_{u \in U^N_s\setminus\{0\}} \sup_{v \in 
V^M\setminus\{0\}} \frac{a(u,v)}{\|u\|_U\|v\|_V}\bigg)^2
\geq (1-\deltaM)c_\trial c_\test\alpha^2.
\end{align*}

Let us now estimate $C_B$. Using the continuity \eqref{eq:continuity} of 
$a(\cdot,\cdot)$ and the Riesz property, we see that
\begin{align*}
\sup_{x\in\sparse\setminus\{0\}} \frac{\|Bx\|_2^2}{\|x\|_2^2}
  &= \sup_{x\in\sparse\setminus\{0\}} \frac{1}{\|x\|_2^2} 
\sum_{j\in[N]}\sum_{k\in[N]} x_j x_k\sum_{q\in[M]} a(\trial_j,\test_q) 
a(\trial_k,\test_q)\\
& \leq \sup_{u\in U^{N}_s\setminus\{0\}} 
\frac{{C_\trial}}{\|u\|_U^2}\sum_{q\in[M]} a(u,\test_q)^2.
\end{align*} 
Notice that for every $u \in U$,
\begin{equation}
\label{eq:bound_bilinear_form}
\sum_{q \in [M]} a(u,\test_q)^2
= \|a(u,\cdot\,) \circ (\Test|_{\mathbb{R}^M})\|_{2}^2
\leq \|a(u,\cdot\, )\|_{V^*}^2 \|\Test|_{\mathbb{R}^M}\|_{\mathbb{R}^M \to V}^2
\leq \|u\|_U^2 \beta^2 C_\test,
\end{equation}
where $\Xi$ is the reconstruction operator associated with the test basis $(\xi_q
)_{q \in \mathbb{N}}$ and where we have used that $\|\Test\|_{\ell^2 \to V}^2 
\leq C_\test$ and that $\|a(u,\cdot\,)\|_{V^*} \leq \|u\|_U \beta$. Combining the 
above inequalities 
yields
$$
\sup_{x\in\sparse\setminus\{0\}} \frac{\|Bx\|_2^2}{\|x\|_2^2} \leq \beta^2 
C_\trial C_\test.
$$

The proof is concluded by applying Theorem~\ref{thm:CORSING_rip} with $c_B = 
(1-\deltaM) c_\trial c_\test\alpha^2$ and $C_B = {C_\trial} C_\test\beta^2$.
\end{proof}

\begin{rem}[Limitations on $\ep$]
\label{rmrk:limitdelta}
Relation \eqref{eq:RIPthm:defdeltam} implies a lower bound for the RIP constant. 
In particular, by letting $\deltaM\to0^+$ in \eqref{eq:RIPthm:defdeltam}, we 
obtain the necessary condition
\begin{equation}
\label{eq:lowerboundlacking}
\ep > 1-  \frac{1}{\condweak}.
\end{equation}
This will imply restrictions on $\condweak$ to guarantee the recovery via OMP, 
studied in the next section. Note that this restriction on the RIP constant
is analogous to the one discussed in Section~\ref{sec:Riesz}.
\end{rem}

\subsection{Recovery via Orthogonal Matching Pursuit}
\label{sec:CORSING_OMP}
We study the performance of the \textsf{CORSING} recovery scheme when the approximate solution 
$\hat{x}$ to \eqref{eq:P0s} is computed via OMP. We choose the OMP algorithm thanks to its parallelizability and
its ability to control the number of iterations, and hence the resulting computational 
cost, when an estimate of the target sparsity is known. In order to prove a precise 
recovery estimate, we define the output of OMP (Algorithm~\ref{alg:OMP}) taking 
into account the $\ell^2$-normalization of the columns (lines \ref{step:norm1} 
and \ref{step:norm2} in Algorithm~\ref{alg:OMP}). 

\begin{algo}(Orthogonal Matching Pursuit)\label{alg:OMP}\\
\textit{Inputs:} $A=(a_1|\cdots|a_N) \in \mathbb{C}^{m \times N}$, $b \in 
\mathbb{C}^{m}$, $s \in \NN$\\
\textit{Output:} $\hat{x} \in \Sigma_{s,N}$\\
\textit{Procedure:} $\hat{x}$ = \texttt{OMP}$(A,y,s)$
\begin{algorithmic}[1]
\STATE $B \gets A R$, with $R_{jk} = \delta_{jk} / \|a_j\|_2$, for 
$j,k \in[N]$; 
\label{step:norm1} 
\STATE $S \gets \emptyset$; $\hat{x} \gets 0$;
\FOR{$i = 1,\ldots,s$}
\STATE $\displaystyle k \gets \arg \max_{j \in [N]} |(B^*(y - B \hat{x}))_j|$;
\STATE $S \gets S \cup \{k\}$;
\STATE $\displaystyle\hat{x} \gets 
\operatorname{argmin}_{z \in \mathbb{C}^N} \|B z -  y\|_2$ 
s.t. $\supp(z) \subseteq S$;
\ENDFOR
\STATE $\hat{x} \gets R \hat{x}$; \label{step:norm2}
\end{algorithmic}
\end{algo}
The main tool employed here is \cite[Theorem 6.25]{holger}, a recovery 
theorem for OMP based on the RIP. This theorem is the generalization of a 
result first published in \cite{Zhang2011}; see also \cite{Cohen2015}.
We give its version in a discrete setting and at the same time add  a slight 
generalization regarding the $\ell^2$-normalization of the columns.
\begin{thm}
\label{thm:OMPrecRIP}
There exist constants $\overline{K}\in\NN$, $C > 0$, and $\ep^* \in (0,1)$ such that for 
every $s \in\NN$, the following holds. If $A\in\mathbb{C}^{m\times N}$ satisfies 
the $\RIP{\ep}{(\overline{K} + 1)s}$, with $\ep < \ep^*$ then, for any $y \in 
\mathbb{C}^m$, the output $\hat{x}= \textnormal{\texttt{OMP}}(A,y,\overline{K} s)$ of 
Algorithm~\ref{alg:OMP} is such that
$$
\|A \hat{x}-y\|_2 \leq C \inf_{z \in \sparse} \| Az-y\|_2.
$$
Possible values for the constants are  $\overline{K}=12$, $C=49$. Moreover, 
$$
\ep^* = 
\begin{cases}
1/6 & \text{if the columns of $A$ are $\ell^2$-normalized},\\
1/13 & \text{otherwise}.
\end{cases}
$$
\end{thm}
\begin{proof}
If $A$ has $\ell^2$-normalized columns, then the theorem is a direct 
consequence of \cite[Theorem 6.25]{holger} (with a minor modification of the proof in order to have   $C\inf_{z \in \Sigma_{s,N}}\|Az-y\|_2$ instead of $C \|Ax_{S^c}-y\|_2$ in the right-hand side of the recovery error bound). The values of the constants $\overline{K} = 
12$, $C = 49$ and $\ep^*=1/6$ are deduced by a direct inspection of the proof (see, in particular, \cite[Proposition 6.24]{holger}).

Let us now assume that there exists some $j\in[N]$, with $\|a_j\|_2 \neq 1$. First, 
observe that if $A$ satisfies the $\RIP{\ep}{(\overline{K}+1)s}$ then $1-\ep \leq 
\|a_j\|_2^2 \leq 1+\ep$, for every $j \in [N]$. Then, considering the 
normalization matrix $R$ defined in step \ref{step:norm1} of 
Algorithm~\ref{alg:OMP}, a direct computation shows that, for every 
$(\overline{K}+1)s$-sparse vector $z$, it holds
$$
\|AR z\|_2^2 
\leq (1+\ep) \|R z\|_2^2 
\leq \frac{1+\ep}{1-\ep} \|z\|_2^2 
= \bigg(1+ \frac{2\ep}{1-\ep}\bigg)\|z\|_2^2,
$$
and, similarly $\|ARz\|_2^2  \geq (1-2\ep/(1+\ep)) \|z\|_2^2$. 
Hence, the matrix $AR$ satisfies the RIP$(2\ep/(1-\ep),(\overline{K}+1)s)$. 
Now, defining $\hat{z} := \texttt{OMP}(AR,y,\overline{K} s)$, we have $\hat{x} = R \hat{z}$. 
Applying \cite[Theorem 6.25]{holger} to $AR$, we finally obtain
$$
     \|A \hat{x} - y\|_2 
   = \|A R \hat{z} - y\|_2 
\leq C \inf_{z \in \sparse} \|A R z - y\|_2 
   = C \inf_{z \in \sparse} \|A z - y\|_2,
$$
provided that the RIP constant of $AR$ is less than or equal to $1/6$. This is 
equivalent to require $2\ep/(1-\ep) < 1/6$, which is equivalent to $\ep < 
1/13=:\ep^*$.
\end{proof}
\begin{rem}
\label{rmrk:constOMP}
An inspection of the argument employed in \cite[Theorem 6.25]{holger} reveals that the 
constants $\overline{K}$, $\OMPerr$, and $\ep^*$ of Theorem~\ref{thm:OMPrecRIP} are  intertwined.  For example, one could relax the condition on $\ep^*$ and 
consequently increase the values of $C$ and $\overline{K}$. 
\end{rem}

Using the notation of Algorithm~\ref{alg:OMP}, we consider the 
\textsf{CORSING} solution defined by 
\begin{equation}
\label{eq:CORSINGrecOMP}
\hat{u}:= \Trial(\texttt{OMP}(DA,Dy, k)) \; ,
\end{equation}
where $D$ is as defined in \eqref{eq:defD}.
Theorem~\ref{thm:CORSINGOMPrecovery} shows that, in order to achieve a 
\textsf{CORSING} error $\|\hat{u}-u\|_U$ comparable to the best 
approximation error of $u$ in $U^N_s$ in expectation, it is sufficient to choose 
$k=\BigO(s)$ iterations of OMP in \eqref{eq:CORSINGrecOMP}. Notice that it is 
possible to show a version of this theorem in probability, 
analogously to \cite{brugiapaglia2018theoretical}.

\begin{thm}[Recovery guarantee for \textsf{CORSING}]
\label{thm:CORSINGOMPrecovery} 
There exist constants $\overline{K} \in \NN$, $\OMPerr > 0$, and $\ep^* \in (0,1)$ 
such that the following holds. Let $\|u\|_U \leq \normubound$ for some 
$\normubound >0$ and assume that the constant $\condweak$ defined in 
\eqref{eq:defzeta} satisfies
\begin{equation}
\label{eq:suffcondCORSING}
\condweak < \frac{1}{1-\ep^*}.
\end{equation}
Let $s,N,M \in \NN$, with $s \leq N$, and 
\begin{equation}
\label{eq:choiceparamCORSINGOMP}
\deltaM \in (0,1-(1-\ep^*) \condweak), \quad
\ep \in \bigg(1- \frac{1-\deltaM}{\condweak},\ep^*\bigg),
\end{equation}
be such that the following truncation condition holds:
\begin{equation*}
\sum_{q > M} \mu_q^N \leq \frac{\alpha^2 \gamma c_\trial c_\test}{(\overline{K} + 1) 
s}.
\end{equation*}
Then, provided that $m$ satisfies \eqref{eq:CORSING_sampling_complexity}, where $\eta = \ep -1 + \frac{1-\gamma}{\condweak}$, the \textsf{CORSING} 
solution $\hat{u}$ computed via OMP as in \eqref{eq:CORSINGrecOMP} satisfies
\begin{equation}
\label{eq:CORSINGOMPerror}
\Expe[\|\mathcal{T}_{\normubound}\hat{u}-u\|_U] 
\leq \bigg(1+\frac{1+\OMPerr}{(1-\ep)^{\frac12}}\bigg) \inf_{w \in 
U^N_s}\|u-w\|_U + 2 \normubound \zeta,
\end{equation}
where $\mathcal{T}_\normubound v := \min\{1,\normubound/\|v\|_U\} v$ and where
$\zeta = 2\exp(-\min\{1,C_\varphi^2C_\xi^2\beta^{4} \}\eta^2 m /(5^{12}s \|\nu^{N,M}\|_1^2))$ 
bounds the failure probability of the RIP. 
Possible values for the constants are $\overline{K} =12$, $\OMPerr = 49$ and $\ep^* = 
1/13$. 
\end{thm}
\begin{proof}
The argument is analogous to that of \cite[Theorem 
3.13]{brugiapaglia2018theoretical}, where the role of the restricted inf-sup 
property is replaced by the RIP. Consider the constants $\overline{K}$, $\OMPerr$, and 
$\ep^*$ as in Theorem~\ref{thm:OMPrecRIP} and define the event 
\begin{equation*}
\Omega_{\text{RIP}}:=\{\tilde{A} \text{ satisfies } \RIP{\ep}{\overline{K}+1)s}\}.
\end{equation*}
We split the expectation accordingly as
$$
\Expe[\|\mathcal{T}_\normubound (\hat u) - u\|_U] = \Expe [ \mathbf{1}_{\Omega_{\text{RIP}}} \| 
\mathcal{T}_\normubound(\hat u) - u\|_U ] +  \Expe[ \mathbf{1}_{\Omega_{\text{RIP}}^C} 
\|\mathcal{T}_\normubound (\hat u) - u \|_U ] \; .
$$
The second term is bounded by $2\normubound\zeta$, since the adopted choice of 
$M$ and $m$ guarantees $\Prob(\Omega_{\text{RIP}}^c) \leq \zeta$, thanks to 
Theorem \ref{thm:RIPcorsing}, and due to the truncation via 
$\mathcal{T}_\normubound$.

Now, consider a generic $w \in U^N_s$. Observing that 
$\mathcal{T}_{\normubound}$ is $1$-Lipschitz with respect to $\|\cdot\|_U$ and 
using the triangle inequality, we find 
\begin{align}
\Expe[\mathbf{1}_{\Omega_{\text{RIP}}} \| \mathcal{T}_\normubound(\hat u) - u\|_U ]
&=\Expe[\mathbf{1}_{\Omega_{\text{RIP}}} \| \mathcal{T}_\normubound(\hat u) - 
\mathcal{T}_\normubound(u)\|_U ] \leq 
\Expe[\mathbf{1}_{\Omega_{\text{RIP}}} \| \hat u - u\|_U ] \nonumber\\
\label{eq:triangle_split}
&\leq \Expe[\mathbf{1}_{\Omega_{\text{RIP}}} \| \hat u - w\|_U ] + 
\Expe[\mathbf{1}_{\Omega_{\text{RIP}}} \| u - w\|_U ].
\end{align}
Notice that the second expectation in \eqref{eq:triangle_split} is trivially bounded by 
$\|u-w\|_U$. In order to bound the other expectation, we note that 
$\hat{x}=\texttt{OMP}(DA,Dy,n) = \texttt{OMP}(\tilde{A},\tilde{y},n)$
and that Theorem~\ref{thm:OMPrecRIP} holds on the event $\Omega_{\text{RIP}}$. Therefore, 
denoting $z = \Trial^* w$, we have the chain of inequalities
\begin{align*}
\|\hat{u} - w\|_U^2 
 & \leq C_\trial \|\hat{x}-z\|_2^2 
   \leq \frac{C_\trial}{1-\ep} \|\tilde{A} (\hat{x}-z)\|_2^2
  \leq \frac{C_\trial}{1-\ep} (\|\tilde{A} \hat{x}-\tilde{y}\|_2 + \|\tilde{A} 
z-\tilde{y}\|_2)^2\\
  &\leq \frac{C_\trial}{1-\ep} (1 + \OMPerr)^2 \|\tilde{A} z-\tilde{y}\|_2^2.
\end{align*} 
Hence, we estimate
\begin{equation}
\label{eq:intuhat-us}
\Expe[\mathbf{1}_{\Omega_{\text{RIP}}} \|\hat{u} - w\|_U ]
\leq \sqrt{\frac{C_\trial}{1-\ep}}(1+\OMPerr) \Expe[\1_{\Omega_{\text{RIP}}} 
\|\tilde{A} z-\tilde{y}\|_2 ].
\end{equation}
Now, exploiting that the $\tau_i$'s are i.i.d., the Riesz property of 
$(\test_q)_{q \in \NN}$ (see the definition in \eqref{eq:Aftilde}), and the continuity \eqref{eq:continuity} of $a(\cdot,\cdot)$, we have
\begin{align*}
\Expe[\|\tilde{A} z - \tilde{y}\|_2^2] 
& = \beta^{-2} C_\trial^{-1} C_\test^{-1} \Expe[\|D (A z - y)\|_2^2]
  = \beta^{-2} C_\trial^{-1} C_\test^{-1}\sum_{q \in [M]} [a(w,\test_q)- 
\Fop(\test_q)]^2\\
&  = \beta^{-2} C_\trial^{-1} C_\test^{-1}\sum_{q \in [M]} a(w-u,\test_q)^2
  \leq C_\trial^{-1} \|w-u\|_U^2.
\end{align*}
Note that we have used inequality \eqref{eq:bound_bilinear_form} in the last step. 
Now, applying Jensen's inequality to the previous relation we obtain 
$\Expe[\|\tilde{A} z - \tilde{y}\|_2] \leq C_\trial^{-1/2} \|w-u\|_U$, which, 
combined with 
\eqref{eq:intuhat-us}, yields
\begin{equation*}
\int_{\Omega_{\text{RIP}}} \|\hat{u} - w\|_U \de\Prob 
\leq \sqrt{\frac{C_\varphi}{1-\ep}} (1+C) \Expe[\|\tilde{A} z - \tilde{y}\|_2]
\leq \frac{1+\OMPerr}{(1-\ep)^{\frac12}} \|u-w\|_U.
\end{equation*}
Combining the above inequalities completes the proof.
\end{proof}

\begin{rem}
Plugging $\ep^* = 1/13$ into relation \eqref{eq:suffcondCORSING}, we obtain 
\begin{equation}
\label{eq:condweakcondition}
\condweak < \frac{13}{12},
\end{equation} 
which is a very restrictive condition. As already mentioned in 
Remark~\ref{rmrk:constOMP}, the value of $\ep^*$ can be made larger, in price of 
larger values of $\overline{K}$ and $\OMPerr$. An interesting open question is whether 
$\ep^*$ can be made arbitrarily close to $1$, or if there is a maximal 
admissible value strictly lower than $1$ (see also the discussion in Section 
\ref{sec:Riesz}).
\end{rem}

We believe that the sufficient condition \eqref{eq:condweakcondition} is too 
conservative. Indeed, numerical experiments show the success of the \textsf{CORSING} 
method in computing accurate sparse approximations via OMP in problems where 
\eqref{eq:condweakcondition} is not met (e.g., advection-dominated problems), see 
\cite{Brugiapaglia2015,simone_thesis,brugiapaglia2018theoretical}. Bridging this 
gap between theory and practice and showing recovery guarantees for OMP without 
assuming \eqref{eq:condweakcondition} is left to future work. 


\section{Proof of Theorem~\ref{prop:main}}

\label{section/proof}
In this section we prove Theorem \ref{prop:main}. Recall that $X
\in \cbb^N$ is a random vector with bounded components, i.e., for all $ i
\leq N :|\inner{X}{e_i}| \leq K$, where $e_1, \dots,e_N$ denotes the standard
basis of $\cbb^N$ and that $X_1,\ldots,X_m$ are independent copies of $X$.
We aim to bound
\begin{equation}
	\sup_{f \in T} \Big| \frac{1}{m} \sum_{i=1}^m |\inner{f}{X_i}|^2 - \ebb
	|\inner{f}{X}|^2 \Big|  \; ,
	\label{eqn:empirical_process}
\end{equation}
for $T \subseteq  \sqrt{s} B_{\ell^1}^N$. Let us recall the following deviation inequality for the
empirical process \cite[Theorem 2.3]{bousquet_concentration_2002}.
\begin{thm} \label{thm:talagrands_inequality}
	Let $\F$ be a class of functions $f: S \to \rbb$ on some set $S$ and let $X_1, \ldots, X_m$ denote
	independent $S$-valued random variables, which are independent copies of random variable $X$. Set
	\begin{equation*}
		Z_\F = \sup_{f \in \F} \Big| \sum_{i=1}^m f(X_i) - \ebb f(X) \Big| \; ,
		\qquad \sigma_\F^2 = \sup_{f \in \F} \ebb f(X)^2 \qquad \text{and} \qquad
		\beta_\F = \sup_{f \in \F} \norm{f}{L^\infty} \; .
	\end{equation*}
	Then, for any $u>0$,
	\begin{equation}
		\pr\Big( Z_\F \geq \ebb Z_\F + \sqrt{2u (\sigma_\F^2 + 2\beta_\F \ebb Z_\F)} + \frac{1}{3} \beta_\F u\Big)
		\leq 2 \exp(-u) \; .
		\label{eq:talagrand}
	\end{equation}
\end{thm}
Theorem~\ref{thm:talagrands_inequality} is related to Talagrand's concentration inequality for the empirical
process.
We employ Theorem~\ref{thm:talagrands_inequality}'s inequality by considering
$\F_T = \{ |\inner{f}{\cdot}|^2: f \in T\}$ for the set $T \subseteq \sqrt{s} B_{\ell^1}^N$.
In this case, for any $u\geq 0$, the event of \eqref{eq:talagrand} reads
\begin{align*}
	\sup_{f \in T} \Big|& \frac{1}{m} \sum_{i=1}^m |\inner{f}{X_i}|^2 - \ebb
	|\inner{f}{X}|^2 \Big|
	\leq  \ebb \sup_{f \in T} \Big|\frac{1}{m} \sum_{i=1}^m
		|\inner{f}{X_i}|^2 - \ebb |\inner{f}{X}|^2\Big| \\
		&\qquad+ \sqrt{2u} \frac{1}{m} \Big( \sigma_{\F_T}^2 m
		+2 \beta_\F \sup_{f \in T} \Big| \sum_{i=1}^m |\inner{f}{X_i}|^2 - \ebb |\inner{f}{X}|^2 \Big| \Big)^{\frac{1}{2}}
		+ \frac{1}{3} \beta_{\F_T} \frac{u}{m} \; .
\end{align*}
Now using the fact that $X$ has bounded coordinates, we obtain that
\begin{equation*}
	\ebb|\inner{f}{X}|^4  \leq s K^2 \sup_{f \in T}
	\ebb |\inner{f}{X}|^2
	\quad\text{and} \quad
	\sup_{f \in T} \norm{|\inner{f}{X}|^2}{L^\infty} \leq K^2 s \; .
\end{equation*}
Hence, if we set $u = \delta^2 m/(K^2 s)$ for $\delta \in (0,1)$, then
Theorem~\ref{thm:talagrands_inequality} yields that with probability at least
$1-2\exp(- \delta^2 m/(K^2 s) )$,
\begin{equation}
	\begin{split}
	\sup_{f \in T} \Big| &\frac{1}{m} \sum_{i=1}^m |\inner{f}{X_i}|^2 - \ebb
	|\inner{f}{X}|^2 \Big| \leq \ebb \sup_{f \in T} \Big|\frac{1}{m} \sum_{i=1}^m
	|\inner{f}{X_i}|^2 - \ebb |\inner{f}{X}|^2 \Big| \\
	&+ \Big( 2\delta^2 \sup_{f \in T} \ebb |\inner{f}{X_i}|^2
		+4 \delta^2 \sup_{f \in T} \Big| \frac{1}{m} \sum_{i=1}^m |\inner{f}{X_i}|^2 - \ebb |\inner{f}{X}|^2 \Big| \Big)^{\frac{1}{2}}
		+ \frac{\delta^2}{3}  \; .
	\end{split}
	\label{eq:concentration_bound}
\end{equation}
It remains to estimate the expectation of the process in \eqref{eqn:empirical_process}.
By symmetrization \cite[Lemma 6.3]{ledoux_probability_1991},
\begin{equation}
	\label{eqn:symmetrization}
	\ebb \sup_{f \in T} \Big|\frac{1}{m} \sum_{i=1}^m
	|\inner{f}{X_i}|^2 - \ebb |\inner{f}{X}|^2 \Big|
	\leq 2 \ebb \sup_{f \in T} \Big|\frac{1}{m} \sum_{i=1}^m
	|\inner{f}{X_i}|^2 \ep_i \Big|  \; ,
\end{equation}
where $(\ep_i)_{i \in[m]}$ denotes a sequence of independent symmetric Bernoulli random variables, also independent of $(X_i)_{i \in [m]}$.
The following theorem provides a bound for the right-hand-side of \eqref{eqn:symmetrization}.
\begin{thm} \label{thm:bound_bernoulli_process}
	Let $T \subseteq \sqrt{s} B_{\ell^1}^N$ and let $\delta \in
	(0,1)$. Then, the following holds:
	\begin{equation*}
	\begin{split}
		 \ebb \sup_{f \in T} \Big|
		 	\frac{1}{m} \sum_{i=1}^m &|\inner{f}{X_i}|^2 \ep_i
			\Big| \\
		&\leq
		(280 + 200 \sqrt{2}) \, \sqrt{ \frac{s K^2 \log^2(sK^2/\delta) \log(eN)}{m} }
		\Big( \ebb \sup_{f \in T} \frac{1}{m} \sum_{i=1}^m |\inner{f}{X_i}|^2 \Big)^{1/2} \\
		&+(69+49\sqrt{2}) \, \delta \sum_{i=1}^m |\inner{f}{X_i}|^2
		+ (643+468\sqrt{2}) \, \delta m \; .
	\end{split}
	\end{equation*}
\end{thm}
Let us observe that this bound implies a bound for the right-hand side of
\eqref{eqn:symmetrization} provided that $\delta \in (0,1)$ is small enough.
\begin{cor} \label{cor:bound_expectation}
	Let $T \subseteq \sqrt{s} B_{\ell^1}^N$ and let $\delta \in (0,\frac{1}{28}(10-7\sqrt{2}))$.
	Assume that
	\begin{equation}
		\label{eqn:measurement_complexity}
		m \geq 1600 (99 + 70 \sqrt{2}) \,  \delta^{-2} s K^2\log(sK^2/\delta)^2 \log(eN) \; .
	\end{equation}
	Then, the following holds:
	\begin{equation*}
		\ebb \sup_{f \in T} \Big|\frac{1}{m} \sum_{i=1}^m
		|\inner{f}{X_i}|^2 - \ebb |\inner{f}{X}|^2 \Big|
		\leq 8(161+117\sqrt{2}) \, \delta
		+ 14(10+7\sqrt{2}) \, \delta \sup_{f \in T} \ebb|\inner{f}{X}|^2\; .
	\end{equation*}
\end{cor}
\begin{proof}
	Let $T \subseteq \sqrt{s} B_{\ell^1}^N$ and $\delta \in (0,1/218)$.
	Choosing $m \geq c\delta^{-2} s K^2\log(sK^2/\delta)^2 \log(eN)$ for an
	appropriate absolute constant $c\geq 1600 (99 + 70 \sqrt{2})$, it follows from
	\eqref{eqn:symmetrization} and Theorem~\ref{thm:bound_bernoulli_process},
	that
\begin{equation*}
	\begin{split}
		\ebb \sup_{f \in T} \Big| \frac{1}{m} \sum_{i=1}^m |\inner{f}{X_i}|^2
		- \ebb |\inner{f}{X_i}|^2 \Big|
		&\leq 2 \delta \, \Big( \ebb \sup_{f \in T} \Big|
		\frac{1}{m} \sum_{i=1}^m |\inner{f}{X_i}|^2 \Big|\Big)^{1/2} \\
		&+ (69+49\sqrt{2}) \, \delta \sum_{i=1}^m |\inner{f}{X_i}|^2
		+ (643+468\sqrt{2}) \, \delta m  \; .
	\end{split}
\end{equation*}
By the arithmetic-mean-geometric-mean inequality we have
\begin{align*}
	\delta \Big( \ebb \sup_{f \in T} \Big|
		\frac{1}{m} \sum_{i=1}^m |\inner{f}{X_i}|^2
		\Big|\Big)^{1/2}
		&= \sqrt{\delta} \Big( \delta \ebb \sup_{f \in T} \Big|
		\frac{1}{m} \sum_{i=1}^m |\inner{f}{X_i}|^2
		\Big|\Big)^{1/2} \\
	&\leq \frac{\delta}{2} + \frac{\delta}{2}
	\ebb \sup_{f \in T} \Big|\frac{1}{m} \sum_{i=1}^m |\inner{f}{X_i}|^2 \Big|\; .
\end{align*}
Combining this with the fact that
\begin{equation*}
	\ebb  \sup_{f \in T}\Big| \frac{1}{m}\sum_{i=1}^m |\inner{f}{X_i}|^2 \Big| \leq
		\ebb \sup_{f \in T} \Big|
		\frac{1}{m} \sum_{i=1}^m |\inner{f}{X_i}|^2 -\ebb |\inner{f}{X}|^2
		\Big| + \sup_{f \in T} \ebb |\inner{f}{X}|^2 \; ,
\end{equation*}
we obtain the inequality
\begin{equation}
	\label{eqn:proof_main_intermediate}
	\begin{split}
	\ebb \sup_{f \in T} &\Big| \frac{1}{m} \sum_{i=1}^m |\inner{f}{X_i}|^2
		- \ebb |\inner{f}{X_i}|^2 \Big| \\
	&\leq  (70+49\sqrt{2}) \, \delta  \ebb \sup_{f \in T} \Big| \frac{1}{m} \sum_{i=1}^m
	|\inner{f}{X_i}|^2 - \ebb|\inner{f}{X}|^2 \Big| \\
	&+ (70+49\sqrt{2}) \delta
	\sup_{f\in T} \ebb |\inner{f}{X}|^2
	+\, (644+468\sqrt{2})\delta
	\; .
\end{split}
\end{equation}
Assuming that $\delta \in (0,\frac{1}{28}(10-7\sqrt{2}))$, we conclude from \eqref{eqn:proof_main_intermediate},
that
\begin{equation*}
		\sup_{f \in T} \Big| \frac{1}{m} \sum_{i=1}^m |\inner{f}{X_i}|^2
		- \ebb |\inner{f}{X_i}|^2 \Big|
		\leq 8(161+117\sqrt{2}) \, \delta
		+ 14(10+7\sqrt{2}) \, \delta \sup_{f \in T} \ebb|\inner{f}{X}|^2 \; .
\end{equation*}
This implies the desired estimate.
\end{proof}
With this estimate all pieces are in place to deduce Theorem~\ref{prop:main}.
\begin{proof}[Proof of Theorem~\ref{prop:main}]
	Combining Corollary~\ref{cor:bound_expectation} with the event~\eqref{eq:concentration_bound},
	it follows that with probability at least $1-2 \exp(-\delta^2 m/(K^2 s))$,
	\begin{equation*}
		\begin{split}
		\sup_{f \in T} &\Big| \frac{1}{m} \sum_{i=1}^m |\inner{f}{X_i}|^2 - \ebb
		|\inner{f}{X}|^2 \Big| \\
		&\leq 8(161+117\sqrt{2}) \delta + 14(10+7\sqrt{2}) \sup_{f \in T} \ebb |\inner{f}{X}|^2  \\
		&+\sqrt{2} \delta \Big( \sup_{f \in T} \ebb |\inner{f}{X_i}|^2
		+ 8(161+117\sqrt{2}) \delta + 14(10+7\sqrt{2}) \delta
		\sup_{f \in T} \ebb |\inner{f}{X}|^2\Big)^{\frac{1}{2}} + \frac{\delta^2}{3} \; .
		\end{split}
	\end{equation*}
	Using the arithmetic-mean-geometric-mean inequality, we find
	\begin{align*}
		\sqrt{2} \delta \Big( &\sup_{f \in T} \ebb |\inner{f}{X_i}|^2
		+ 8(161+117\sqrt{2}) \delta + 14(10+7\sqrt{2}) \delta \sup_{f \in T} \ebb |\inner{f}{X}|^2\Big)^{\frac{1}{2}} \\
		&\leq \delta + \frac{\delta}{2}\sup_{f \in T} \ebb |\inner{f}{X_i}|^2
		+8(161+117\sqrt{2}) \frac{\delta^2}{2} + 14(10+7\sqrt{2}) \frac{\delta^2}{2} \sup_{f \in T} \ebb |\inner{f}{X}|^2
	\end{align*}
	Hence, on the same event and using the fact that $\delta <1$,
	\begin{equation*}
		\sup_{f \in T} \Big| \frac{1}{m} \sum_{i=1}^m |\inner{f}{X_i}|^2 - \ebb
		|\inner{f}{X}|^2 \Big|
		\leq  \frac{1}{6}(1457+1053\sqrt{2}) \delta +  \frac{1}{2}(421+294\sqrt{2}) \delta \sup_{f \in T} \ebb |\inner{f}{X}|^2 \; .
	\end{equation*}
	This implies the theorem, Since $\frac{1}{2}(421+294\sqrt{2})<\frac{1}{6}(1457+1053\sqrt{2})<492$.
\end{proof}

We are left with establishing Theorem \ref{thm:bound_bernoulli_process}. The
proof of this theorem will occupy the rest of this section. In order to
establish Theorem \ref{thm:bound_bernoulli_process} we will start by
introducing a general quantity, which can control the
Bernoulli process. This quantity is a mixture of an $\ell^1$
approximation term and a finite precision approximation of a
$\gamma_2$ functional (the  meaning of the term ``finite precision'' will
be clarified in Section~\ref{subsec:Bounding_Bernoulli}).

\subsection{Generic chaining}
Let $(T,d)$ denote a (semi-)metric space.
An increasing  sequence $(A_n)_{n \geq 0}$ of {subsets of $T$} is called
admissible if, for all $n \geq 0$,
$|A_n| \leq 2^{2^n}$. For a set $A \subseteq T$ we set $d(A,x) = \inf_{a \in A}
d(a,x)$.
A central object of study in generic chaining are the functionals
\begin{equation}
	\gamma_{\alpha}(T,d):= \inf_{\A} \sup_{t \in T} \sum_{n\geq 1}
2^{\frac{n}{\alpha}}
	d(A_n,t) \; ,
	\label{eqn:definition_gamma}
\end{equation}
where the infimum is taken over all admissible sequences $\A=(A_n)_{n \geq 0}$ {of subsets of} $T$.
\subsection{Bounding the Bernoulli process}
\label{subsec:Bounding_Bernoulli}
In this section we consider subsets $U$ of the space $\rbb^m$
{equipped with the Euclidean distance}
and study bounds for the Bernoulli process given by
\begin{equation}
	\sup_{x \in U} \sum_{i=1}^m x_i \ep_i\; ,
	\label{eqn:definition_bernoulli_process}
\end{equation}
where, as before, $(\ep_i)_{i \leq m}$ is a sequence of independent symmetric
Bernoulli random variables. The set $U$ will later be {replaced by} the set of sequences
$\{(|\inner{f}{X_i}|^2)_{i \leq m}
: f \in {T}\}$ for some $T \subseteq \sqrt{s} B_{\ell^1}^N$.  With the intention of keeping the necessary
notation as simple as possible, we formulate the results of this section in terms
of subsets of $\rbb^m$.

Hoeffding's inequality implies that increments of the process $\sup_{x \in U}
|\sum_{i=1}^m
x_i \ep_i|$ are subgaussian with respect to the Euclidean metric on $\rbb^m$.
Combining
this observation with the trivial bound $\sup_{x \in U} |\sum_{i=1}^m x_i
\ep_i|
\leq \sup_{x \in U} \norm{x}{1}$ and a standard generic
chaining argument yields the bound (see \cite{talagrand_upper_2014}, in
particular
the discussion at the beginning of chapter 5)
\begin{equation}
	\ebb \sup_{x \in U} \sum_{i=1}^m x_i \ep_i \leq b(U) := \inf_{U \subseteq U_1
+ U_2}
	\Big\{\sup_{x \in U_1} \norm{x}{1} + \gamma_2(U_2, \norm{\cdot}{2})
\Big\} \; .
	\label{eqn:generic_chaining_bernoulli}
\end{equation}
For our result, we are only interested in a finite precision approximation for
the Bernoulli process $\sup_{x \in U} | \sum_{i=1}^m x_i \ep_i|$. We therefore
propose to substitute  $b(U)$ by a finite precision version, which is able to control
the left hand side of \eqref{eqn:generic_chaining_bernoulli}.
The following result is implicitly stated in \cite{talagrand_upper_2014}.
A proof is provided for the convenience of the reader.

\begin{lem}\label{lem:bernoulli_expectation}
	Let $U \subseteq \rbb^m$. Let $n_0,\ell>0$, let $(A_n)_{n \geq n_0}$
	denote any admissible sequence for $U$ and let $\pi_n: U \to A_n$
	be generic maps for $n \geq n_0$. Moreover, let $\pi_n(x) = 0$ for
	every $x \in U$, for $n < n_0$. Then, there is a constant $c({n_0}) \in (1,2)$
	such that
	\begin{equation}
		\ebb \sup_{x \in U} \Big| \sum_{i=1}^m x_i \ep_i \Big|
		\leq \sup_{x \in U} \norm{x -
		\pi_{n_0+\ell}(x)}{1}
		+  c(n_0)\, \sup_{x \in U} \sum_{n=n_0}^{n_0+ \ell} 2^{\frac{n}{2}}
		\norm{\pi_n(x)-\pi_{n-1}(x)}{2} \; .
		\label{eqn:expectation_bound_b}
	\end{equation}
\end{lem}
\begin{proof}
	For $n < n_0$ we set $\pi_n(x) = 0$. Fix $\ell>0$ and observe that
	\begin{align*}
	\Big|\sum_{i=1}^m x_i \ep_i \Big| &= \Big| \sum_{i=1}^m \ep_i (x_i -
\pi_{n_0+\ell}(x)_i)
	+ \sum_{n=n_0}^{n_0+\ell} \sum_{i=1}^m \ep_i (\pi_n(x) - \pi_{n-1}(x))_i
\Big| \\
	\leq&
	\norm{x - \pi_{n_0+\ell}(x)}{1}
	+ \sum_{n=n_0}^{n_0+\ell} \Big| \sum_{i=1}^m \ep_i (\pi_n(x) -
\pi_{n-1}(x))_i \Big| \; .
	\end{align*}
	For $x \in \rbb^m$ let
	$Z_x = \sum_{i=1}^m \ep_i x_i$ denote the associated random variable. Taking suprema on
	each side in the inequality above yields
	\begin{equation}
		\sup_{x \in U} \Big| \sum_{i=1}^m \ep_i x_i \Big|
		\leq \sup_{x \in U} \norm{x - \pi_{n_0 + \ell} (x)}{1} +
		\sup_{x \in U} \sum_{n=n_0}^{n_0 + \ell} |Z_{\pi_n(x)} -
		Z_{\pi_{n-1}(x)}|\; .
		\label{eqn:basic_chaining}
	\end{equation}
	By Hoeffding's inequality (see, e.g., \cite{boucheron_concentration_2013})
	we obtain the estimate
	\begin{equation}
		\pr\Big( |Z_{\pi_n(x)} - Z_{\pi_{n-1}(x)}| > 2^{\frac{n}{2}} u
		\norm{\pi_n(x)-\pi_{n-1}(x)}{2}
		\Big) \leq 2 e^{-2^{n-1} u^2}
		\label{eqn:probability_bernoulli}
	\end{equation}
	Since $(A_n)_{n \geq n_0}$ is admissible there are at most $|A_n||A_{n-1}|
	\leq 2^{2^{n+1}}$
	pairs of the form $(\pi_n(x),\pi_{n-1}(x))$ and it follows by a union bound
	that the event $\Omega_n(u)$
	\begin{equation}
		\begin{split}
		\forall &(\pi_n(x),\pi_{n-1}(x)) \in A_n \times A_{n-1}:
		|Z_{\pi_n(x)} - Z_{\pi_{n-1}(x)}| \leq  2^{\frac{n}{2}} u
		\|\pi_n(x)-\pi_{n-1}(x)\|_2
		\end{split}
		\label{eqn:chaining_good_event}
	\end{equation}
	occurs with probabilty at least $1-2 \cdot 2^{2^{n+1}} e^{-2^{n-1} u^2}$.
	Therefore, a union bound over $n= n_0,\dots,n_0+\ell$ yields, that the probability
	that the event $\Omega(u) =
	\bigcup_{n=n_0}^{n_0+\ell} \Omega_n(u)$ does not occur is bounded by
	\begin{equation}
	    2\sum_{n = n_0}^{n_0 + \ell} 2^{2^{n+1}} e^{-2^{n-1} u^2}
	    \leq 2 e^{-u^2/2} \sum_{n=n_0}^{n_0+\ell} 2^{2^{n+1}} e^{-2^{n+1}}
	    \leq 2 \kappa(n_0) e^{-u^2/2} \; ,
	\end{equation}
	where $\kappa(n_0)=\sum_{n=n_0}^{n_0+\ell} (2/e)^n$ and where we used
	that $u^2 2^{n-1} \geq u^2/2 + u^2 2^{n-2} \geq u^2/2 + 2^{n+1}$ holds
	for $n \geq 1$ and $u \geq 1$. On the event $\Omega(u)$, the inequality \eqref{eqn:basic_chaining} 	reads
	\begin{equation*}
		\sup_{x \in U} |Z_x|
		\leq \sup_{x \in U} \norm{x - \pi_{n_0+\ell}(x)}{1} + u \sup_{x \in U}
		\sum_{n=n_0}^{n_0+\ell}
		2^{\frac{n}{2}} \norm{\pi_n(x) - \pi_{n-1}(x)}{2} \; .
	\end{equation*}
	This implies that the following tail bound is valid for $\sup_{x \in U} |Z_x|$, $n_0\geq 1$
	and $u\geq 1$,
	\begin{equation}
	\label{eq:tail_bound_bernoulli_process}
	   \pr\Big( \sup_{x \in U} |Z_x| >  \sup_{x \in U} \norm{x - \pi_{n_0+\ell}(x)}{1}
	   + u \sup_{x \in U} \sum_{n=n_0}^{n_0+\ell}
		2^{\frac{n}{2}} \norm{\pi_n(x) - \pi_{n-1}(x)}{2} \Big)
		\leq 2 \kappa(n_0) \, e^{-u^2/2} \; .
	\end{equation}
	Set $d_1 = \sup_{x \in U}\|x-\pi_{n_0+\ell}(x)\|_1$ and $d_2=\sup_{x \in U} \sum_{n=n_0}^{n_0+\ell}	2^{\frac{n}{2}} \norm{\pi_n(x) - \pi_{n-1}(x)}{2} $. Then, integrating the tail bound \eqref{eq:tail_bound_bernoulli_process} yields
	\begin{align*}
		\ebb\left[\sup_{x \in T} |Z_x|\right]
		& = \int_{0}^\infty \pr\left(\sup_{x \in T} |Z_x| \geq t \right) dt
		 \leq d_1 + d_2 \int_0^{\infty} \pr\left(\sup_{x \in U} |Z_x| \geq d_1 +
		 u d_2 \right) du\\
		& \leq d_1 + \left(2 \kappa(n_0) \int_1^{\infty}e^{-u^2/2}du + 1\right) d_2 \\
		& \leq d_1 + (\frac{4}{5} \kappa(n_0) + 1) d_2,
	\end{align*}
	where we used the change of variable $t = d_1 + u d_2$. This, shows that the desired estimate is
	true for $c(n_0) = \frac{4}{5} \kappa(n_0) + 1 >0$. Now, obeserve that by the definition of $\kappa(n_0)$,
	\begin{equation*}
	    \kappa(n_0) \leq \left( \frac{2}{e} \right)^{2^{n_0+1}} \sum_{n=0}^\infty \left(
	    \frac{2}{e} \right)^{n} = \left( \frac{2}{e} \right)^{2^{n_0+1}} \frac{e}{e-2}
	    \leq \frac{111}{100} \; .
	\end{equation*}
	This finishes the proof.

\end{proof}

This lemma leads us to introduce the following definition.
\begin{defn}
Let $U \subseteq \rbb^m$ and let $\A$ denote the set of all admissible sequences {$(A_n)_{n \geq n_0}$}
in $U$. For each admissible sequence let $\pi_n : U \to A_n$ denote
a generic map.
For $n_0, \ell>0$ we define
\begin{equation*}
	b_{n_0,\ell}(T) : = \inf_{\A} \Big\{\sup_{x \in T}
		\norm{x - \pi_{n_0+\ell}(x)}{1}
		+ c(n_0) \, \sup_{x \in T}\sum_{n=n_0}^{n_0+\ell}
	2^{\frac{n}{2}} \norm{\pi_n(x)-\pi_{n-1}(x)}{2}\Big\} \; .
\end{equation*}
\end{defn}
In order to find a bound for $b_{n_0,\ell}(U)$ and hence for $\sup_{x \in U} \left|
\sum_{i=1}^m x_i \ep_i\right|$ given a fixed set $U\subseteq \rbb^m$ it suffices to identify a
suitable admissible sequence $(A_n)_{n \geq n_0}$.
\subsection{Estimating $b_{n_0,\ell}(T)$}
\label{subsec:estimating_b}
Returning to our initial setting we are looking for estimates of
$$\sup_{f \in T} \sum_{i=1}^m |\inner{f}{X_i}|^2 \ep_i \; ,$$
for a set $T \subseteq \sqrt{s} B_{\ell^1}^N$. By Lemma \ref{lem:bernoulli_expectation}
it suffices to  bound $b_{n_0, \ell}(\{ (|\inner{f}{X_i}|^2)_{i \leq m} : f
\in T \})$.
Such bounds can be achieved by constructing a suitable admissible sequence
$(A_n)_{n \geq n_0}$ in the set $\{ (|\inner{f}{X_i}|^2)_{i \leq m} : f
\in T \}$. In this section we provide such an admissible
sequence for $b_{n_0,\ell}(\{ (|\inner{f}{X_i}|^2)_{i \leq m} : f
\in T \})$ based on an inital covering of the
set $\{ (\inner{f}{X_i})_{i \in [m]} : f \in T\}$ with respect to the
(semi-)norms $\norm{x}{I}$ defined below.
The construction is based on the notion of a weak covering of a set $
T\subseteq \cbb^N$, which is explained in the following subsection.

\subsubsection{Weak coverings}
We start by defining the empirical seminorms
$$
\norm{x}{I,\X} := \max_{i \in [m]\setminus I} |\inner{X_i}{x}|,
$$
where $I
\subseteq [m]$ and $\X = (X_1, \dots,X_m)$ denotes a realization of the random
vectors in question. Let $B_{I,\X} := \{ x \in \cbb^N : \norm{x}{I,\X} \leq 1\}$
denote the unit balls with respect to the seminorm $\norm{\cdot}{I,\X}$.
In the following we will be interested in covering a set $T \subseteq \cbb^N$ by
the following family of sets for a given width $\rho>0$,
\begin{equation}
 \B(\rho,M) := \{ \rho B_{I,\X} : I \subseteq [m], |I| \leq M \} \; .
\end{equation}
\begin{defn}
	Let $\rho, M > 0$. We say that $T \subseteq \cbb^N$ is \textit{weakly covered} by
	$\B(\rho,M)$, if for some $r \in \mathbb{N}$, there exits $x_1,\ldots,x_r \in \cbb^N$
	and sets $B_1, \dots, B_r \in \B(\rho,M)$, such that
	\begin{equation}
		\label{eqn:def_covering}  T \subseteq \bigcup_{i=1}^r (x_i + B_i) \; .
	\end{equation}
    We call the smallest $r \in \mathbb{N}$ such that \eqref{eqn:def_covering} is
	satisfied the \textit{weak covering number} of $T$ and denote this number by
	$ \N^*(T,\rho,k)$. Moreover, we refer to the set
	$\{x_1,,\ldots,x_r\}$ as \emph{weak covering} or \emph{net}.
\end{defn}
For sets $T \subseteq \sqrt{s} B_{\ell^1}^N \subset \cbb^N$, Maurey's empirical method can be used
to estimate the size of a weak covering of $T$.
\begin{lem}[Maurey's empirical method] \label{lem:maurey}
	Let $T \subseteq \sqrt{s} B_{\ell^1}^N$, $\delta \in
	(0,1)$, $\rho>0$ and $X_1, \ldots, X_m$ denote random vectors with bounded
	coordinates. Then,
	\begin{equation}
		\log \N^*\Big(T,\rho,\frac{4\delta m}{s K^2\log_2(sK^2/\delta) }\Big)
		\leq
		\frac{2 \log_2(sK^2/\delta) \log_2(2N) s K^2}{\rho^2}
		 \; .
		\label{eqn:covering_number_estimate}
	\end{equation}
\end{lem}
\begin{proof}
	We observe that every $x \in T \subseteq \sqrt{s} B_{\ell^1}^N$ is a
	convex combination of $V = \{\pm \sqrt{s} e_j \, \pm \sqrt{s} i e_j: j \in [N]\}$, i.e.\ there
	is a sequence $(\lambda_v)_{v \in V}$ with $\lambda_v \geq 0$ and $\sum_{v \in V} \lambda_v = 1$,
	such that $\sum_{v \in V} \lambda_v v=x$. For a fixed $x \in T \subseteq \sqrt{s} B_{\ell^1}^N$
	the associated sequence $\lambda = (\lambda_v)_{v \in V}$
	defines a probability distribution on $V$. Let $Z$ denote a random vector in $V$
	with distribution $\lambda$ and let $Z_l$ for $l=1,\ldots,L$ denote independent
	copies of $Z$. Set $L = \lfloor K^2 s \rho^{-2} \log_2(s K^2\log_2(sK^2/\delta) /\delta)
	\rfloor$ and consider the random set
	\begin{equation*}
		 J = \Big\{ i \in [m] : \Big| \inner{X_i}{x} - \frac{1}{L}
		\sum_{l \leq L} \inner{X_i}{Z_l} \Big| > \rho \Big\} \; .
	\end{equation*}
	For this set we will bound the probability of the event, that $|J|
	\geq \delta m/ (sK^2\log_2(sK^2/\delta))$. Recall that for $l=1,\dots,L$ we have
	$\ebb_{Z_l} \inner{X_i}{Z_l} = \inner{X_i}{x}$. By Hoeffding's inequality
	it follows that for every $i \in [m]$
	\begin{equation*}
		\pr_{Z_1,\ldots, Z_L}
		\Big(i \in J
		\Big) \leq 2\exp(- L \rho^2 /(s\norm{X_i}{\infty}^2)) \leq
		\frac{2\delta }{s K^2 \log_2(sK^2/\delta)} \; ,
	\end{equation*}
	where we used that $|\langle X_i, Z_l\rangle| \leq \sqrt{s}\|X_i\|_\infty
	\leq \sqrt{s}K$ for every $i \in[m]$ and $l \in [L]$.
	Hence,
	we have $\ebb_{Z_1,\ldots,Z_L} |J| \leq \frac{2 \delta m}{sK^2 \log_2(sK^2/\delta)}$.
	Using Chebyshev's inequality it follows that
	\begin{equation*}
		 \pr_{Z_1,\ldots,Z_L}\Big( |J| \geq \frac{4 \delta m}{s \log(sK^2/\delta)}\Big)
		 \leq \frac{\ebb_{Z_1,\ldots,Z_L} |J|
		s  \log_2(sK^2/\delta)}{4 \delta m} \leq \frac{1}{2} \; .
	\end{equation*}
	Therefore, we can find a realization of $Z_1,\ldots,Z_L$, such that $|J| \leq
	\frac{4 \delta m}{s \log_2(s/\delta)}$. For this realization we have
	\begin{equation*}
		\Big\|x-L^{-1}\sum_{l\leq L}Z_l\Big\|_{J,\mathbf{X}}
		= \max_{i \in [m]\setminus J} \Big| \Big \langle X_i, L^{-1}
		\sum_{l \leq L} Z_l - x \Big\rangle \Big| \leq \rho \; .
	\end{equation*}
	Since $x \in T$ is arbitrary, by considering all possible realizations of
	$L^{-1}\sum_{l\leq L}Z_l$ we obtain a weak covering of $T$. For the choice of
	$L = \lfloor \rho^{-2} s K^2 \log_2(s K^2 \log_2(sK^2/\delta)
	/\delta) \rfloor \leq 2 \lfloor \rho^{-2} s K^2 \log_2(s K^2 /\delta)\rfloor$
	we have at most $|V|^L =  (2N)^L \leq (2N)^{2\rho^{-2} K^2 s \log_2(sK^2/\delta)}$
	realizations of $Z_1,\ldots,Z_L$. This implies the desired result.
\end{proof}

Based on this estimate we derive the following result, which provides the
promised admissible sequence.
\begin{thm}\label{thm:bound_the_metric}
	Let $T \subseteq \sqrt{s} B_{\ell^1}^N$ and let $\delta \in (0,1)$.
	For $n_0 = \lceil \log_2\log_2(2N) + \log_2\log_2(sK^2/\delta) \rceil$ and $\ell = \lceil \log_2(sK^2/\delta)
	\rceil$	there is an admissible sequence $(A_n)_{n \geq n_0}$ for $T$, such that
	\begin{itemize}
		\item For all $f \in T$ we have
			\begin{equation}
			\label{eqn:l1_error_bound}
			\begin{split}
				&\norm{(|\inner{f}{X_i}|^2)_{i \in [m]}
					-\pi_{n_0+\ell}
				(|\inner{f}{X_i}|^2)_{i \in [m]}}{1} \\
				&\quad\qquad\leq (69+49\sqrt{2}) \, \delta \sum_{i=1}^m |\inner{f}{X_i}|^2
		+ (643+468\sqrt{2}) \, \delta m
				\; .
			\end{split}
			\end{equation}

		\item For all $f \in T$ we have
		\begin{equation}
			\begin{split}
				\sum_{n = n_0}^{n_0 + \ell - 1}
				2^{\frac{n}{2}} &\norm{ (|\inner{\pi_{n+1}(f)}{X_i}|^2)_{i \in [m]}
					-
				(|\inner{\pi_{n}(f)}{X_i}|^2)_{i \in [m]}}{2} \\
				&\qquad\qquad\qquad\leq
	   (200 + 140 \sqrt{2}) \, \sqrt{\ell s} \, K
	   2^{\frac{n_0}{2}} \Big(\sum_{i=1}^m |\inner{f}{X_i}| ^2\Big)^{1/2}  .
			\end{split}
			\label{eqn:distance_bound}
		\end{equation}
	\end{itemize}
\end{thm}
Theorem \ref{thm:bound_the_metric} directly implies a bound for $b_{n_0,\ell}(
\{ (|\inner{f}{X_i}|^2)_{i \in [m]} : f \in T\})$ with $T \subseteq \sqrt{s}
B_{\ell^1}^N$ and henceforth it implies a bound for the expectation of
\eqref{eqn:empirical_process}. We remark that the size of $n_0$ is determined
by the initial estimate given in Lemma \ref{lem:maurey},
while the parameter $\ell$ depends on the desired approximation accuracy
$\delta \in (0,1)$  and the $\ell^\infty$-diameter of the set $\{ (\inner{f}{X_i})_{i \in [m]}
:f \in T\}$.
\subsubsection{Construction of  an admissible sequence and proof of Theorem~\ref{thm:bound_the_metric}}
The goal of this section is ultimately to bound $b_{n_0,\ell}(T)$ for
a subset $T \subset \sqrt{s} B_{\ell^1}^N$, thereby identifying $n_0, \ell >0$.
The plan is to achieve a bound for $b_{n_0,\ell}(T)$ by  means of constructing an
admissible sequence $(A_n)_{n \geq n_0}$ (see \eqref{eq:definition_admissible_set}
below). The main point here is to ensure that the admissible sequence
$(A_n)_{n \geq n_0}$ balances the bounds for the terms
\begin{equation*}
    \sup_{x \in T} \norm{x - \pi_{n_0 + \ell}(x)}{1} \qquad \text{and} \qquad
    \sup_{x \in T} \sum_{n=n_0}^{n_0+\ell-1} 2^{\frac{n}{2}} \norm{\pi_{n+1}(x) - \pi_{n}(x)}{2} \; .
\end{equation*}
The approach of this section is based on ideas in \cite{bou:improved_rip,DBLP:journals/corr/HavivR15a}. The argument we propose is based on the observation that we have to deal with sequences $(|\inner{f}{X_i}|^2)_{i \in [m]}$ for $f \in T$
and aims at organizing the admissible sequence $(A_n)_{n \geq n_0}$ around this fact. A vital observation is that for a fixed $i \in [m]$ we have
\begin{equation*}
    \Big| |\inner{f}{X_i}|^2 - |\inner{g}{X_i}|^2 \Big| \leq 2
    \max \{ |\inner{f}{X_i}|, |\inner{g}{X_i}| \}
    |\inner{f}{X_i} - \inner{g}{X_i}| \; .
\end{equation*}
We leverage this by coupling the approximation of $|\inner{f}{X_i}|^2$ in the $\ell^\infty$-norm to the size of the coefficients of
the sequence $(|\inner{f}{X_i}|^2)_{i \in [m]}$. This gives us
a quasi-$\ell^\infty$-control along the chain. This $\ell^\infty$-control allows us to simultaneously bound the differences $\norm{\pi_{n+1}(x) - \pi_{n}(x)}{2}$ and, if $\ell$ is sufficiently large, to bound $\norm{x - \pi_{n_0+\ell}(x)}{1}$.

Henceforth, we will use the following notation. For an element
$f \in T \subseteq \sqrt{s} B_{\ell^1}^N$ and a realization of $X_1, \ldots, X_m$ we write
$f_{\mathbf{X}} = (\inner{f}{X_i})_{i \in [m]}$ and $|f_\X|^2 =
(|\inner{f}{X_i}|^2)_{i \in [m]}$. Further, for the rest of this section we fix
an approximation accuracy $\delta \in (0,1)$. Let $k_0$ denote a positive
integer satisfying $k_0 = \lceil \log_2(1/\delta^2) \rceil$.
Moreover, we let $\ell = \lceil \log_2(sK^2/\delta) \rceil$ and
$n_0 = \lceil \log_2\log_2(2N)+\log_2\log_2(sK^2/\delta) \rceil$ as in the setting of
Theorem~\ref{thm:bound_the_metric}.

Let us start by observing that Lemma \ref{lem:maurey} shows that for $T \subset \sqrt{s} B_{\ell^1}^N$
there is a sequence of weak coverings for the parametrized families of sets $\B(\rho_n,M)$ with
\begin{equation}
	\rho_n = \sqrt{s} K 2^{-n/2} \quad \text{for} \quad
	n = 0, \ldots, \ell + k_0 \; ,
	\qquad \text{and}
	\qquad
	M = \frac{4 \delta m}{sK^2 \log_2(sK^2/\delta)} \; .
	\label{eq:definition_rho}
\end{equation}
We denote the associated nets as $\widetilde{A}_{n_0+n}$, for $n = 0,\ldots, \ell+k_0$.
By \eqref{eqn:covering_number_estimate}
and since $n_0 \geq \log_2\log_2(2N) + \log_2\log_2(sK^2/\delta)$, the weak nets
$\widetilde{A}_{n_0}, \ldots, \widetilde{A}_{n_0+\ell+k_0}$ corresponding to $\rho_{0},
\ldots, \rho_{\ell + k_0}$ and $M$ as in \eqref{eq:definition_rho}
satisfy
\begin{equation*}
	|\widetilde{A}_{n_0+n}|
	\leq 2^{\log_2(2N) \log_2(sK^2/\delta) K^2 s / \rho_n^2}
	= 2^{\log_2(2N) \log_2(sK^2/\delta) \cdot 2^n}
	\leq 2^{2^{n_0+n}},
\end{equation*}
and are therefore admissible sequences in $T$.
For each weak covering $A_{n_0+n} = \{x_1,\ldots,x_r\}$ any point $x_i$ for $i=1,\ldots,r$ comes with
its own unitball. Hence, each point in $x_i$ for $i=1,\ldots,r$ is only correctly
described by the tuple $(x_i,\norm{\cdot}{\X,I_i})$ for $i = 1,\ldots,r$. Let us
therefore agree on the following notation.
For each $f \in T$ we denote by $\widetilde{\pi}_{n_0+n}(f)$ any element $x \in \widetilde{A}_{n_0+n}$
of the weak net $\widetilde{A}_{n_0 + n} =\{x_1,\ldots,x_r\}$ that satisfies
$$
\|x-\widetilde{\pi}_{n_0+n}(f)\|_{\widetilde{I}_{n_0+n}(f), \X}
= \min_{i=1,\dots,r} \norm{f - x_i}{I_i,\X},
$$
for a suitable $\widetilde{I}_{n_0+n}(f) \subseteq [m]$ with $|\widetilde{I}_{n_0+n}(f)|\leq M$ such
that the pair $(\widetilde{\pi}_{n_0+n}(f),\norm{\cdot}{\widetilde{I}_{n_0+n}(f)})$ is an element of
$\{(x_i,\norm{\cdot}{\widetilde{I}_{n_0+n}(f))})\}_{i \in [r]}$.
By the same token we also introduce another piece of
notation, which is necessary because of the nature of the seminorms $\norm{\cdot}{\X,I}$.
As pointed out above, by the definition of
$\norm{\cdot}{\X,I}$ each approximant $\widetilde{\pi}_{n_0 + n}(f)$ is associated with
a set $I$ satisfying $|I| \leq M$, such that $|\inner{f-\widetilde{\pi}_{n_0 + n}(f)}{X_i}| >
\rho_n$ for all $i \in I$. Let us therefore define,
\begin{equation}
\label{eqn:def_failure_sets}
I_{n_0+n}(f) := \{i \in [m] :  |\inner{f-\widetilde{\pi}_{n_0 + n}(f)}{X_i}| >
\rho_n \} \; .
\end{equation}
We denote the sequence $\widetilde{\pi}_{n_0 }(f), \ldots,
\widetilde{\pi}_{n_0 + n}(f)$ by $\I_n(f_\X)$.

We refine the admissible sequence $(\widetilde{A}_{n_0 + n})_{n=0}^{\ell+k_0}$
based on our chosen approximation accuracy $\delta>0$.
For each $f \in T$ and some $n \in \mathbb{N}$ we consider the tuple
$(\widetilde{\pi}_{n_0}(f), \dots, \widetilde{\pi}_{n_0+n}(f))$ and inductively define
sets $(E_{n_0+n})_{n = 0}^{\ell}$ as $E_{n_0-1} := \emptyset$  and, for $n
= 0, \ldots, \ell$, as
\begin{equation} \label{eq/def/sets}
	E_{n_0+n} = E_{n_0+n}(f_\X) := \{ i \in [m] \; : \;
	|\inner{\widetilde{\pi}_{n_0+n}(f)}{X_i}|
	\geq \sqrt{2} \cdot \rho_{n}\}
 \setminus \bigcup_{k < n} E_{n_0 + k}(f_\X) \; .
\end{equation}
We define an admissible sequence of approximations for an element $f \in T $
based on the data $(\widetilde{\pi}_{n_0 +n}(f))_{n=0}^{\ell + k_0}$
and the associated sets $(E_{n_0+n}(f_\X))_{n =0}^{\ell}$. For each sequence $|f_\X|^2$, we set
\begin{equation} \label{eqn:definition_admissible_parts}
	\pi_{n,k}^*|f_\X|^2 :=
	\begin{cases}
	(|\inner{\widetilde{\pi}_{n_0+k}(f)}{X_i}|^2
	\textbf{1}_{E_{n_0+n}}(i))_{i \in [m]} \; ,  & \text{if } k \in [n,n+k_0],\\
	0, & \text{if } k \not \in[n,n+k_0].
	\end{cases}
\end{equation}
We also introduce the notation
$\pi_{n,k}^*|\inner{f}{X_i}|^2$ for the $i$-th component of
$\pi_{n,k}^*|f_\X|^2$. Given these maps and using this shorthand notation, we
set
\begin{equation}
\label{eq:definition_differences}
 \pi_{k,l}^\#|\inner{f}{X_i}|^2 :=
 \begin{cases}
  \pi_{k,l}^*|\inner{f}{X_i}|^2 - \pi_{k,l-1}^*|\inner{f}{X_i}|^2 =:\theta,  & \text{if }|\theta| \leq
  (11+9\sqrt{2}) \,\rho_{k} \rho_l,\\
0, & \text{otherwise,}
\end{cases}
\end{equation}
for every $i \in [m]$ and
\begin{equation}
\label{eqn:definition_admissible}
\begin{split}
	\pi_{n_0+n+1}|\inner{f}{X_i}|^2
	:= \sum_{k \leq n} \sum_{l \in [k,k+k_0]}
 	\pi_{k,l}^\#|\inner{f}{X_i}|^2 \; ,
\end{split}
\end{equation}
for all $i \in [m]$. The associated sequence of sets $A_{n_0+n+1}$ is the range of
the maps $\pi_{n_0+n+1}: T \to A_{n_0+n+1}$ for $n = 0, \ldots, \ell-1$, i.e.
\begin{equation}
	\label{eq:definition_admissible_set}
 A_{n_0+n+1} := \Big\{ \Big(\sum_{k \leq n} \sum_{l \in [k,k+k_0]}
 \pi_{k,l}^\#|\inner{f}{X_i}|^2\Big)_{i \in [m]} : f \in T \Big\} \; .
\end{equation}
Moreover, for $n = -1$, note that $\pi_{n_0}|f_{\X}| = 0$ and $A_{n_0}= \{0\}$.
To show that the sequence of sets $(A_{n_0+n})_{n=0}^{\ell-1}$ we need to provide a
bound for the size of $A_{n_0+n}$ for each $n=0,\dots,\ell-1$. We recall that by construction
$|\widetilde{A}_n| \leq 2^{2^n}$ provided that $n \geq n_0$. Moreover, for each
$n \in [0,\ell]$ and every $f \in T$ we have that the elements in $A_{n_0+n}$
are determined by the first $n$ elements
$\widetilde{\pi}_{n_0}(f),\dots,\widetilde{\pi}_{n_0+n-1}(f)$. Hence,
	\begin{equation*}
		|A_{n_0+n}| \leq \prod_{k \leq n-1} |\widetilde{A}_{n_0+k}| \leq 2^{\sum_{k \leq n-1}
		2^k} \leq 2^{2^{n} - 1} \leq 2^{2^n}  \; .
	\end{equation*}
Therefore, the sequence $(A_n)_{n=n_0}^{n_0 +\ell}$ is an admissible
sequence for $T$. Before we turn to the proof of Theorem~\ref{thm:bound_the_metric} let us observe the following
fact concerning the $\ell^\infty$-norm of the approximants, which will come handy in proofing Theorem~\ref{thm:bound_the_metric}.
For each $f \in T$ the definition \eqref{eq:definition_differences}
implies that,
\begin{equation}
\label{eq:l_infty_bound_app}
\begin{split}
    \max_{i \in [m]} \pi_{n_0+n}|\inner{f}{X_i}|^2  &\leq \sum_{k = 0}^{n-1} \sum_{l \in [k,k+k_0]}
 	\pi_{k,l}^\#|\inner{f}{X_i}|^2 \\
 	&\leq (11+9\sqrt{2}) \sqrt{s} K^2 \sum_{k \leq n-1} 2^{-k} \sum_{l=0}^{k_0} 2^{-l/2} \\
 	&\leq (11+9\sqrt{2}) \sqrt{s} K^2 \sum_{k=0}^\infty 2^{-k} \sum_{l=0}^{\infty} 2^{-l/2} \\
 	&\leq (80+58\sqrt{2}) \sqrt{s} K^2 \; .
\end{split}
\end{equation}

Now that these estimates are in place we are left with ensuring the properties
of this sequence claimed by Theorem \ref{thm:bound_the_metric}. We have split
this task into several lemmas, which can be combined to three leading
principles for the admissible seqeunce, \textit{$\ell^2$-stability}, \textit{$\gamma_2$-boundedness} and \textit{$\ell^1$-approximation}.
\begin{itemize}
    \item \textit{$\ell^2$-stability.} Lemma~\ref{lem:bound_g} contains the observation that for $i \in E_{n_0 +n}$
    the approximants $\inner{\widetilde{\pi}_{n_0+n}(f)}{X_i}$ have size roughly equal to $\rho_n$.
    In line with this observation is Lemma~\ref{lem:approximation_level_sets}, which states
    that the sequence $(|E_{n_0+n}|\rho_n^2)_{n=0}^\ell$ captures essentially the $\ell^2$-norm
    of $(\inner{f}{X_i})_{i \in [m]}$.
    \item \textit{$\gamma_2$-boundedness.} Lemma~\ref{lem:Euclidean_distance_bound} provides the key bound for $\norm{\pi_{n_0+n+1}|f_\X|^2 - \pi_{n_0+n} |f_\X|^2 }{2}$ and therefore a bound for the rightmost
    term in $b_{n_0,\ell}(T)$.
    \item \textit{$\ell^1$-approximation.} Lemma~\ref{lem:approximation} finally provides a bound for the $\ell^1$-approximation
    term and is the last step towards proving Theorem~\ref{thm:bound_the_metric}.
\end{itemize}

We start with the observation regarding the interplay between the approximants
$\widetilde{\pi}_{n_0+n}(f)$ and the sets $E_{n_0+n}$.
\begin{lem} \label{lem:bound_g}
	Let $\ell > 0$, $f \in T \subseteq \sqrt{s} B_{\ell^1}^N$  and let $X_1,\ldots,X_m$ denote
	realizations of $X$. Let $\pi_{n+1} : T \to A_{n+1}$ be defined as in
	\eqref{eqn:definition_admissible}, \eqref{eq:definition_differences}, and \eqref{eqn:definition_admissible_parts}
	and let $E_{n_0},\ldots,E_{n_0+\ell}$ denote the associated sets defined as in \eqref{eq/def/sets}. Then,
	for all $n=0,\ldots,\ell$, $k \in [n,n+k_0]$ and for all
	$i \in [m] \setminus \bigcup_{n \leq \ell} I_{n_0+n}(f)$ we have
	$\1_{E_{n_0+n}}(i)|\inner{\widetilde{\pi}_{n_0+k}(f)}{X_i}| \leq
	(3 + \sqrt{2}) \rho_n$.
\end{lem}
\begin{proof}
We first observe that if $n=0$ the
definition of the nets $\widetilde{A}_{n_0+n}$
implies that for each $i \in [m]$, $|\inner{\widetilde{\pi}_{n_0}(f)}{X_i} |
\leq \norm{\widetilde{\pi}_{n_0}(f)}{1} K \leq \rho_0$
and, if $n \geq 1,$  the definition \eqref{eq/def/sets} of the set $E_{n_0+n}$ implies that
$\1_{E_{n_0+n}}(i) |\inner{\widetilde{\pi}_{n_0+n-1}(f)}{X_i} | \leq \sqrt{2} \cdot
\rho_{n-1}$.
Further, by definition \eqref{eqn:def_failure_sets} of the sets $I_{n_0+n}(f)$, for all
$i \in [m]$ outside of the set $\bigcup_{n \leq \ell} I_{n_0+n}(f)$ the
estimates $|\;|\inner{\widetilde{\pi}_{n_0+k}(f)}{X_i}| -
|\inner{f}{X_i}|\;| \leq \rho_k$ and
$|\;|\inner{\widetilde{\pi}_{n_0+n-1}(f)}{X_i}| -|\inner{f}{X_i}|\;| \leq
\rho_{n-1}$ hold simultaneously .
Recalling the definition \eqref{eq:definition_rho} of $\rho_n$, it follows that
\begin{equation} \label{eq/lem/bound_g/ind_step}
\begin{split}
	|\inner{\widetilde{\pi}_{n_0+k}(f)}{X_i}| \1_{E_{n_0 +n}}(i)
	&\leq
	\1_{E_{n_0+n}}(i) (|\inner{f}{X_i}| + \rho_k) \\
	&\leq
	\1_{E_{n_0 +n}}(i)( |\inner{\widetilde{\pi}_{n_0+n-1}(f)}{X_i}| +
		\rho_{n-1} + \rho_k) \\
	&\leq
	( 3+ \sqrt{2}) \rho_n
\, .
\end{split}
\end{equation}
 This concludes the proof.
\end{proof}
Let us further observe that by a similar argument as in the lemma we find that
for each $i \in [m] \setminus \bigcup_{n \leq \ell} I_{n_0+n}(f)$, and for each $n=0,\ldots, \ell$
we have the bounds
\begin{align*}
    \sqrt{2} \rho_n \leq \1_{E_{n_0+n}}(i) |\inner{\widetilde{\pi}_{n_0+n}(f)}{X_i}| \leq \1_{E_{n_0+n}}(i) |\inner{f}{X_i}| + \rho_n
\end{align*}
and
\begin{align*}
    \1_{E_{n_0+n}}(i) |\inner{f}{X_i}| \leq \1_{E_{n_0+n}}(i) |\inner{\widetilde{\pi}_{n_0+n}(f)}{X_i}| + \rho_n
    \leq (4 + \sqrt{2}) \rho_n\; .
\end{align*}
In conjunction with Lemma~\ref{lem:bound_g} these bounds imply that for
$i \in [m] \setminus \bigcup_{n \leq \ell} I_{n_0+n}(f)$ and $n=0,\ldots,\ell$,
\begin{itemize}
	\item[\useequationref{eqn:coeff_bound_app}]
		if $i \in E_{n_0 + n}$, then for all $k \in [n,n+k_0]$,
			$
			|\inner{\widetilde{\pi}_{n_0+ k}(f)}{X_i}|^2
			\leq (3+\sqrt{2})^2 \rho_n^2 \; ,
			$
	\item[\useequationref{eqn:coeff_bound}] if $i \in E_{n_0 + n}$, then
			$(\sqrt{2} - 1)^2 \rho_n^2 \leq |\inner{f}{X_i}|^2
			\leq (\sqrt{2} + 4)^2 \rho_n^2 \; .
			$
\end{itemize}

\begin{lem} \label{lem:Euclidean_distance_bound}
	Let $(A_n)_{n \geq n_0}$ be as in \eqref{eq:definition_admissible_set},
	$f \in T \subseteq \sqrt{s} B_{\ell^1}^N$ and
	$E_{n_0},\dots, E_{n_0+\ell}$ defined as in \eqref{eq/def/sets}. Then,
	\begin{equation*}
		\norm{ \pi_{n_0+n+1} |f_\X|^2
			-\pi_{n_0 + n} |f_\X|^2 }{2} \\
			\leq  (40 + 29 \sqrt{2}) \, \sqrt{|E_{n_0+n}|} \rho_n^2
		\; .
	\end{equation*}
\end{lem}
\begin{proof}
Let $f \in T$. Expanding the short-hand notation $|f_\X|^2$ and the $\ell^2$ norm, we have
	\begin{equation}
		\norm{ \pi_{n_0 + n+1} |f_\X|^2 - \pi_{n_0 + n}|f_\X|^2 }{2}
		= \Big( \sum_{i =1}^m (\pi_{n_0+n+1}|\inner{f}{X_i}|^2 - \pi_{n_0+n}
		|\inner{f}{X_i}|^2 )^2 \Big)^{1/2}.
	\label{eq:lem47_computation_1}
	\end{equation}
Recalling the definition \eqref{eqn:definition_admissible} of $\pi_n |\inner{f}{X_i}|^2$ for each $i \in [m]$, we find that the
right hand side of \eqref{eq:lem47_computation_1} is equal to
\begin{align}
\Big(
\sum_{i = 1}^m &
\Big(
\sum_{k \leq n} \sum_{l
\in [k,k+k_0]}\pi_{k,l}^\# |\inner{f}{X_i}|^2 -
\sum_{k \leq n-1} \sum_{l \in [k,k+k_0]}
\pi_{k,l}^\#|\inner{f}{X_i}|^2
\Big)^2
\Big)^{1/2} \nonumber\\
\label{eq:proof_lemma_euclidean_distance_1}
&= \Big( \sum_{i = 1}^m \Big( \sum_{l \in [n,n+k_0]}
\pi_{n,l}^\#|\inner{f}{X_i}|^2 \Big)^2 \Big)^{1/2} \; .
\end{align}
By definition \eqref{eq:definition_differences} and recalling \eqref{eqn:definition_admissible_parts},
we see that the support of the approximant $(\pi_{n,l}^\#|\inner{f}{X_i}|^2)_{i \in [m]}$ is contained
in the intersection of $E_{n_0+n}$ and the set
\begin{align*}
 J_{n,l}  := \{ i \in [m]  : | |\inner{\widetilde{\pi}_{n_0 + l}(f)}{X_i}|^2
 - |\inner{\widetilde{\pi}_{n_0+l-1}(f)}{X_i}|^2| \leq (11+9\sqrt{2})\, \rho_{n}\rho_l \} \; .
\end{align*}
Hence, we can rewrite the above sum in \eqref{eq:proof_lemma_euclidean_distance_1} as
\begin{align*}
\Big(
	\sum_{i = 1}^m ( \1_{E_{n_0+n}}(i) &\sum_{l \in [n,n+k_0]} \1_{J_{n,l}}(i)
    (|\inner{\widetilde{\pi}_{n_0 + l}(f)}{X_i}|^2
    - |\inner{\widetilde{\pi}_{n_0+l-1}(f)}{X_i}|^2)^2 \Big)^{1/2}
\\
&\leq (11+9\sqrt{2}) \; \Big( \sum_{i = 1}^m    \Big( \1_{E_{n_0+n}}(i)
			\rho_n \sum_{l \in [n,n+k_0]} \1_{J_{n,l}}(i)\rho_{l} \Big)^2 \Big)^{1/2}
        \end{align*}
	Recalling that $\rho_l = K \sqrt{s} 2^{-l/2}$ this implies
\begin{align*}
	\sum_{l \in [n,n+k_0]} \rho_l
	&= K \sqrt{s} \sum_{l \in [n,n+k_0]} 2^{-l/2}
	= K \sqrt{s} 2^{-n/2} \sum_{l \in [0,k_0]} 2^{-l/2} \\
	&\leq K \sqrt{s} 2^{-n/2} \sum_{l = 0}^\infty 2^{-l/2}
	\leq \frac{\sqrt{2}}{\sqrt{2}-1} K \sqrt{s} 2^{-n/2}.
\end{align*}
Hence, the right hand side of \eqref{eq:lem47_computation_1} is bounded by
\begin{align*}
	(40 + 29 \sqrt{2}) \, \rho_n \Big( \sum_{i=1}^m  \1_{E_{n_0+n}}(i) (2 \underbrace{K \sqrt{s} 2^{-n/2}}_{=\rho_n} )^2
	\Big)^{1/2}
	\leq (40 + 29 \sqrt{2}) \, \sqrt{|E_{n_0+n}|} \rho_n^2 \; .
\end{align*}
Combining the above inequalities, we obtain the desired estimate.
\end{proof}

\begin{lem} \label{lem:approximation_level_sets}
	Let $E_{n_0}, \ldots, E_{n_0+\ell}$ be
	the sets associated with $\pi_{n_0+\ell} |f_\X|^2$, defined as in \eqref{eq/def/sets}.
	Then, for every realization of $X_1, \ldots, X_m$ and for every
	$f \in T \subseteq \sqrt{s} B_{\ell^1}^N$, we have
	\begin{equation}
		(\sqrt{2}-1)^2 \sum_{n \leq \ell} |E_{n_0+n}| \rho_{n}^2 \leq \sum_{i=1}^m
		|\inner{f}{X_i}|^2 \; .
	\end{equation}
\end{lem}
\begin{proof}
	The sets $(E_{n_0 + n})_{n = 0}^{\ell}$
	together with the set $E_* :=	\{ i \in [m] :
	|\inner{\widetilde{\pi}_{n_0+\ell}(f)}{X_i}| < \sqrt{2} \rho_{\ell} \}$ form
	a partition of the set $[m]$. Let $\I_\ell(f_\X) = (I_{n_0}(f), \dots, I_{n_0 +\ell}(f))$ denote the sets
	associated with the initial approximants defined in \eqref{eqn:def_failure_sets}. Then, by definitions
	\eqref{eq/def/sets} and \eqref{eqn:def_failure_sets}, and using the fact that the sets $(E_{n_0+n})_{n= 0}^{\ell}$ are disjoint, for every $i \in [m]
	\setminus \bigcup_{n \leq \ell} I_{n_0 + n}(f)$, we see that
	\begin{align*}
		|\inner{f}{X_i}| &\geq  \sum_{n \leq \ell} |\inner{f}{X_i}|
		\1_{E_{n_0 +n}}(i)
		\\
		&\geq \sum_{n \leq \ell}
			|\inner{\widetilde{\pi}_{n_0+n}(f)}{X_i}| \1_{E_{n_0 + n}}(i)
		- \rho_n \1_{E_{n_0 + n}}(i)
		\\
		&\geq \sum_{n \leq \ell} (\sqrt{2} \rho_{n} - \rho_n) \1_{E_{n_0 + n}}(i)
		\geq (\sqrt{2}-1) \sum_{n \leq \ell} \rho_{n} \1_{E_{n_0 + n}}(i)
		\; .
	\end{align*}
	Squaring both sides of the inequality above, summing over all indices $i \in [m] \setminus \bigcup_{n \leq \ell} I_{n_0 + n}(f)$
	and, again, using the fact that the sets $E_{n_0+n}$ are disjoint, it follows that
	\begin{equation} \label{eq:lower_bound_1}
	\sum_{i=1}^m |\inner{f}{X_i}|^2 \geq
		\sum_{i \in [m] \setminus \bigcup_{n \leq \ell} I_{n_0 + n}(f)} |\inner{f}{X_i}|^2
		\geq  (\sqrt{2}-1)^2\sum_{n \leq \ell} \rho_{n}^2 |E_{n_0 + n}|
		\; .
	\end{equation}
	This is the desired estimate.
\end{proof}
Next we study the problem of bounding the difference $\norm{  |f_\X|^2 -
\pi_{n_0 + \ell} |f_\X|^2}{1}$. The following lemma provides a key step towards
Theorem \ref{thm:bound_the_metric}.
\begin{lem} \label{lem:approximation}
	Let $T \subseteq \sqrt{s} B_{\ell^1}^N$ and let $\pi_{n_0+\ell}: T \to
	A_{n_0+\ell}$ be as defined in \eqref{eqn:definition_admissible}.
	Further, let $n_0 \geq \log_2\log_2(2N) + \log_2\log_2(sK^2/\delta)$ with $\delta \in (0,1)$,
	$\ell = \lceil \log_2(sK^2/\delta) \rceil$, $k_0 = \lceil \log_2(1/\delta^2) \rceil$.
	Then, for every realization of $X_1, \dots, X_m$ and every $f \in T$ we have
	\begin{equation}
	\label{eq:l1_error_bound}
	 	\sum_{i=1}^m \Big|\pi_{n_0+\ell}| \inner{f}{X_i}|^2 - |\inner{f}{X_i}|^2\Big|
    		\leq
		(69+49\sqrt{2}) \, \delta \sum_{i=1}^m |\inner{f}{X_i}|^2
		+ (643+468\sqrt{2}) \, \delta m \; .
	\end{equation}
\end{lem}
\begin{proof}
	For the approximation $\pi_{n_0+\ell}|\inner{f}{X_i}|^2$ we let
	$\widetilde{\pi}_{n_0}(f),\dots,\widetilde{\pi}_{n_0+\ell}(f)$ denote the
	inital approximants
	in $T$ and $E_{n_0},\dots,E_{n_0+\ell}$ the corresponding sets of indices
	defined in \eqref{eq/def/sets}.
	Futher, let $\I_\ell(f_\X) = (I_{n_0}(f),\dots,I_{n_0 +\ell}(f))$ denote the sets
	associated with the approximants, defined by \eqref{eqn:def_failure_sets}.
    Moreover, we define the sets
	$$
	J := [m] \setminus
	\bigcup_{n \leq \ell} I_{n_0 + n}(f) \quad \text{and} \quad J^c = \bigcup_{n \leq \ell} I_{n_0 + n}(f).
	$$
	We split the proof into two main parts. First, we show that the
	contribution from the set $J^c$ to the approximation error, i.e., $\norm{(|f_\X|^2 -
\pi_{n_0 + \ell} |f_\X|^2)|_{J^C}}{1}$, is bounded by
	$16 (40+29\sqrt{2})\, \delta m$.  Afterwards, we study the contributions of the set $J$, namely $\norm{(|f_\X|^2 -
\pi_{n_0 + \ell} |f_\X|^2)|_{J}}{1}$.

	For the first part, recalling from \eqref{eq:definition_rho} that $M = 4\delta m /(sK^2 \log_2(sK^2/\delta))$,
	we observe that by the definition of weak covering and by the definition \eqref{eqn:def_failure_sets} of the
	sets $I_{n_0 + n}$ we have $|I_{n_0+n}(f)| \leq M$ and, therefore,
	\begin{equation}
		|J^c|
		\leq (\ell +1) \max_{n \leq \ell} |I_{n_0+n}(f)|
		\leq \frac{4\delta  m (\ell+1)}{s K^2 \log_2(sK^2/\delta)} \; .
		\label{eq:bound_size_bad_sets}
	\end{equation}
	Hence, using \eqref{eq:l_infty_bound_app} we find that for each
	$f \in T \subseteq \sqrt{s} B_{\ell^1}^N$,
	\begin{equation}
		\begin{split}
		\sum_{i \in J^c} &\Big| |\inner{f}{X_i}|^2
		- \pi_{n_0+\ell} |\inner{f}{X_i}|^2\Big|  \\
		&\leq \Big( \max_{i \in [m]}
		|\inner{f}{X_i}|^2 + \max_{i \in [m]}\pi_{n_0+\ell} |\inner{f}{X_i}|^2
		\Big)
		\frac{4 m (\ell+1) \delta}{sK^2\log_2(sK^2/\delta)} \\
		&\leq 16 (40+29\sqrt{2}) \, \delta m \; ,
		\end{split}
		\label{eqn:absolute_estimate}
	\end{equation}
	In \eqref{eqn:absolute_estimate} we have used that $\ell = \lceil \log_2(sK^2/\delta) \rceil$.
	Therefore, the contribution from the indices $i \in J^c$ is bounded by $16 (40+29\sqrt{2}) \, \delta m$.

	In the second part of the proof, we estimate the approximation error corresponding to the set $J$. We start by expressing the approximant $\pi_{n_0+\ell}|f_{\X}|^2$ in an equivalent form (see \eqref{eqn:def_equivalent_approximant} below) obtained by simplifying a telescoping sum. Namely, for every $i \in J$, we observe that,
	assuming that
	\begin{equation}
	\label{eq:intermetiate_claim}
	    \1_{E_{n_0+n}}(i)(|\pi_{n,k}^*|\inner{f}{X_i}|^2-
	\pi_{n,k-1}^*|\inner{f}{X_i}|^2) \leq (11+9\sqrt{2}) \cdot \rho_{n} \rho_k,
	\end{equation}
	for every $k \in [n,n+k_0]$ and $n \in[0,\ell]$, then for each $i \in [m]$
	the telescoping sum in the definition \eqref{eqn:definition_admissible} satisfies
	\begin{equation}
	\label{eqn:def_equivalent_approximant}
	\begin{split}
 		\pi_{n_0+\ell}|\inner{f}{X_i}|^2
 		&= \sum_{n=0}^{\ell-1}
 			\sum_{k \in [n,n+k_0]} (|\inner{\pi_{n,k}^*(f)}{X_i}|^2 -
			 |\inner{\pi_{n,k-1}^*(f)}{X_i}|^2) \\
	    	&= \sum_{n=0}^{\ell-1} |\inner{\pi_{n,n+k_0}^*f}{X_i}|^2
		= \sum_{n=0}^{\ell-1} \1_{E_{n_0+n}}(i)
		|\inner{\widetilde{\pi}_{n_0+n+k_0}(f)}{X_i}|^2  \; .
	\end{split}
	\end{equation}
	Let us now verify the validity of  \eqref{eq:intermetiate_claim}. This is trivial if $i \notin E_{n_0+n}$. For $i \in E_{n_0+n}$,
	it is implied by the following estimate: For each $i \in J$ and $k \in [n,n+k_0]$,
	we have
	\begin{equation} \label{eq:distance_X_to_parts}
		\Big | |\inner{f}{X_i}|^2 - \pi_{n,k}^* |\inner{f}{X_i}|^2 \Big|
		= \Big | |\inner{f}{X_i}|^2
		- \widetilde{\pi}_{n_0+k}|\inner{f}{X_i}|^2 \Big|
		\1_{E_{n_0+n}}(i)
		\leq (7+ 2 \sqrt{2}) \, \rho_{k} \rho_n\; .
	\end{equation}
	Let us show the validity of \eqref{eq:distance_X_to_parts} for $i \in E_{n_0+n}$.
	Recall that for $i \in J$ and $k \in [n,n+k_0]$ the definition of \eqref{eq:definition_rho}
	implies that $|\;|\inner{f}{X_i}| - |\inner{\widetilde{\pi}_{n_0+k}(f)}{X_i}| \;|
	\leq \rho_k$. Therefore, by using \eqref{eqn:coeff_bound_app}  and \eqref{eqn:coeff_bound} we find,
	\begin{align*}
		\Big| |\inner{f}{X_i}|^2 &- |\inner{\widetilde{\pi}_{n_0+k}(f)}{X_i}|^2 \Big| \1_{E_{n_0+n}}(i) \\
		&=  \Big||\inner{f}{X_i}| - |\inner{\widetilde{\pi}_{n_0+k}(f)}{X_i}| \Big| \cdot
		\Big||\inner{f}{X_i}| + |\inner{\widetilde{\pi}_{n_0+k}(f)}{X_i}| \Big| \1_{E_{n_0+n}}(i)\\
		&\leq \rho_k \Big|  |\inner{f}{X_i}| + |\inner{\widetilde{\pi}_{n_0+k}(f)}{X_i}| \Big|  \1_{E_{n_0+n}}(i) \\
		&\leq (7+ 2 \sqrt{2}) \rho_k \rho_n
	\end{align*}
	This proves the claim \eqref{eq:distance_X_to_parts}. With this we observe that
	for each $i \in J \cap E_{n_0+n}$ we have
	\begin{equation} \label{eq/proof/lem/main/final_bound}
	| \pi_{n,k}^* |\inner{f}{X_i}|^2 - \pi_{n,k-1}^* |\inner{f}{X_i}|^2| \leq
(11+ 9 \sqrt{2}) \, \rho_{n} \rho_k \;.
	\end{equation}
	Indeed, for $k
	\geq n+1$ we employ \eqref{eq:distance_X_to_parts} and the triangle
	inequality to see that
	\begin{align*}
		| \pi_{n,k}^* |\inner{f}{X_i}|^2  - \pi_{n,k-1}^* & |\inner{f}{X_i}|^2| \\
		&= \Big| \pi_{n,k}^* |\inner{f}{X_i}|^2-
		|\inner{f}{X_i}|^2
		+ |\inner{f}{X_i}|^2- \pi_{n,k-1}^* |\inner{f}{X_i}|^2\Big|  \\
		&\leq \Big| \pi_{n,k}^* |\inner{f}{X_i}|^2-
		|\inner{f}{X_i}|^2\Big|+ \Big| \pi_{n,k-1}^* |\inner{f}{X_i}|^2-
		|\inner{f}{X_i}|^2\Big|  \\
		&\leq ((7+ 2 \sqrt{2}) \rho_{n} \rho_k + (7+ 2 \sqrt{2}) \rho_{n} \rho_{k-1}) \\
		&\leq (7+ 2 \sqrt{2}) (1+\sqrt{2}) \,  \rho_{n} \rho_k  = (11 + 9 \sqrt{2}) \,  \rho_{n} \rho_k  \; .
	\end{align*}
	Further, for $k = n$ the desired estimate follows from Lemma
	\ref{lem:bound_g} and for $k<n$ we have $\pi_{n,k}^* |\inner{f}{X_i}|^2 = 0$.
	This proves \eqref{eq/proof/lem/main/final_bound} and, consequently,
	\eqref{eq:intermetiate_claim} and, in turn, \eqref{eqn:def_equivalent_approximant}.

	Taking advantage of the representation \eqref{eqn:def_equivalent_approximant}, the proof is now concluded by estimating $\norm{(|f_\X|^2 -
\phi_{n_0,\ell} |f_\X|^2)|_{J}}{1}$, where we define
	\begin{equation}
		 \phi_{n_0,\ell}|\inner{f}{X_i}|^2 := \sum_{n=0}^{\ell-1}
		 \1_{E_{n_0+n}}(i) |\inner{\widetilde{\pi}_{n_0+n+k_0}(f)}{X_i}|^2 \; .
	\end{equation}
	For any $i \in J$, there are two possibilities. Either $i$ belongs to a
	set $E_{n_0+n}$ for exactly one value of $n = 0, \ldots, \ell$ (recall that the sets $(E_{n_0+n})_{n \leq \ell}$ are disjoint), or $i$ does not belong to any of these
	sets. In the first case, i.e., $i \in E_{n_0 + n} \cap J$ for some $n=0,\dots, \ell$, the
	estimate \eqref{eq:distance_X_to_parts} applied for $k = n + k_0$ together with
	\eqref{eqn:coeff_bound}, i.e., $(\sqrt{2} - 1)^2 \rho_n^2 \leq |\inner{f}{X_i}|^2$
	and the fact that $k_0 = \lceil \log_2(1/\delta^2) \rceil$ imply that
	\begin{align*}
		\Big| |\inner{\widetilde{\pi}_{n_0 + n + k_0}(f)}{X_i}|^2 - |\inner{f}{X_i}|^2\Big|
		&\leq  (11+9\sqrt{2}) \, \rho_{n} \rho_{n + k_0} \\
		&\leq (11+9\sqrt{2}) \,\rho_{n}^2 2^{-k_0/2}
		\leq (69+49\sqrt{2}) \, \delta |\inner{f}{X_i}|^2,
	\end{align*}
	Hence, since the sets $(E_{n_0+n})_{n \leq \ell}$ are disjoint, we obtain the estimate
	\begin{equation} \label{eq:good_sets_coverd}
		\sum_{i \in J \cap \bigcup_{n \leq \ell} E_{n+n_0}}
		\Big| |\inner{f}{X_i}|^2 - \phi_{n_0,\ell}|\inner{f}{X_i}|^2
		\Big| \leq (69+49\sqrt{2})\, \delta \sum_{i=1}^m |\inner{f}{X_i}|^2 \; .
	\end{equation}
	Next, let us consider the case that for all $n \leq \ell$ we have $i \not
	\in E_{n_0 + n}$. Under this condition, for all $n \leq \ell$ we have
	$\phi_{n_0,\ell}|\inner{f}{X_i}|^2 = 0$.
	From the definition of $E_{n_0 + n}$ and the fact that $\widetilde{\pi}_{n_0+\ell}(f)$
	is associated with a weak covering of parameter $\rho_\ell$, it follows that, for any $i \in J$,
	\begin{equation*}
		|\inner{f}{X_i}|
		\leq \rho_{\ell} + |\inner{\widetilde{\pi}_{n_0+\ell}(f)}{X_i}|
		\leq \rho_{\ell}  + \sqrt{2} \rho_{\ell}
		\leq (1+\sqrt{2}) \rho_{\ell} \leq (1+\sqrt{2}) \sqrt{\delta} \; .
	\end{equation*}
	Therefore, if $i \not \in E_n$ for all $n \leq \ell$, we obtain
	\begin{equation} \label{eq:good_sets_not_covered}
		\begin{split}
		\sum_{i \in J \cap \bigcap_{n \leq \ell} E_{n_0+n}^c}
		\Big|\phi_{n_0,\ell} |\inner{f}{X_i}|^2 - |\inner{f}{X_i}|^2 \Big|
		&= \sum_{i \in J \cap \bigcap_{n \leq \ell} E_{n_0+n}^c}
		|\inner{f}{X_i}|^2 \\
		&\leq (1+\sqrt{2})^2 \delta \; | J \cap \bigcap_{n \leq \ell} E_{n_0+n}^c| \\
		&\leq (1+\sqrt{2})^2 m \delta \; .
	\end{split}
	\end{equation}
	Finally, since $[m] = (J \cap \bigcup_{n \leq \ell} E_{n_0+n}) \cup (J \cap \bigcap_{n \leq \ell}
	E_{n_0+n}^c) \cup J^c$ is a partition of $[m]$, combining the estimates \eqref{eq:good_sets_coverd},
	\eqref{eq:good_sets_not_covered} and \eqref{eqn:absolute_estimate} we obtain,
	\begin{align*}
		\sum_{i=1}^m \Big| &\pi_{n_0+\ell}|\inner{f}{X_i}|^2 - |\inner{f}{X_i}|^2
		\Big| \\
		&= \sum_{i \in J \cap \bigcup_{n \leq \ell} E_{n_0+n}}
		\Big| \pi_{n_0+\ell}|\inner{f}{X_i}|^2 - |\inner{f}{X_i}|^2 \Big| \\
		&\qquad+ \sum_{i \in J \cap \bigcap_{n \leq \ell}  E_{n_0+n}^C}
		\Big| \pi_{n_0+\ell}|\inner{f}{X_i}|^2 - |\inner{f}{X_i}|^2 \Big| \\
		&\qquad + \sum_{i \in J^C}
		\Big| \pi_{n_0+\ell}|\inner{f}{X_i}|^2 - |\inner{f}{X_i}|^2
		\Big| \\
		&\leq (69+49\sqrt{2})\, \delta \sum_{i \in J \cap \bigcup_{n \leq \ell} E_{n_0+n}}
			|\inner{f}{X_i}|^2
			+ (1+\sqrt{2})^2 \delta m + 16 (40+29\sqrt{2}) \, \delta m \; .
	\end{align*}
	This estimate finishes the proof.
\end{proof}
Having established these results we are ready to give the proof of Theorem
\ref{thm:bound_the_metric}.
\begin{proof}[Proof of Theorem \ref{thm:bound_the_metric}]
	The inequality \eqref{eqn:l1_error_bound} coincides with the claim of Lemma~\ref{lem:approximation}.

Let us verify the second claim \eqref{eqn:distance_bound} of
Theorem~\ref{thm:bound_the_metric}.
	By Lemma \ref{lem:Euclidean_distance_bound} we have the bound
		\begin{equation*}
			\sum_{n = n_0}^{n_0 + \ell-1}
				2^{\frac{n}{2}} \norm{ \pi_{n_0+n+1} |f_\X|^2 - \pi_{n_0+n} |f_\X|^2 }{2}
			\leq
			(40+29\sqrt{2}) \,\sum_{n = n_0}^{n_0 + \ell-1}
				2^{\frac{n}{2}} \sqrt{|E_{n_0+n}|} \rho_n^2
	\end{equation*}
	Recalling that $\rho_n = \sqrt{s} K 2^{-\frac{n}{2}}$, we see that
	the left hand side of \eqref{eqn:distance_bound} is bounded by
	\begin{equation*}
	(40+29\sqrt{2})
	 \,\sum_{n = n_0}^{n_0 + \ell-1}
		2^{\frac{n}{2}} \sqrt{|E_{n_0+n}|} \rho_n^2
		\leq
		(40+29\sqrt{2}) \sqrt{s}K 2^{\frac{n_0}{2}}
		\sum_{n = 0}^{\ell-1} \sqrt{|E_{n_0+n}|} \rho_n  \; .
	\end{equation*}
	By using Cauchy-Schwarz to estimate the sum over $n= 0,\ldots,\ell-1$, it
	follows that this term is bounded by
	\begin{equation}
		\label{eqn:bound_gamma_2_step_1}
		(40+29\sqrt{2}) \cdot 2^{\frac{n_0}{2}} \sqrt{\ell s} K
		 \Big(\sum_{n = 0}^{\ell-1} |E_{n_0+n}| \rho_{n}^2 \Big)^{1/2}.
	\end{equation}
	Applying Lemma \ref{lem:approximation_level_sets}
	to this estimate we obtain that \eqref{eqn:bound_gamma_2_step_1} is bounded by
   \begin{equation*}
	   (200 + 140 \sqrt{2}) \, \sqrt{\ell s} \, K
	   2^{\frac{n_0}{2}} \Big(\sum_{i=1}^m |\inner{f}{X_i}| ^2\Big)^{1/2} \; .
   \end{equation*}
   Summarizing these estimates we obtain \eqref{eqn:distance_bound}.
\end{proof}
The last step in this section is to establish
Theorem~\ref{thm:bound_bernoulli_process} based on Theorem~\ref{thm:bound_the_metric}.
\begin{proof}[Proof of Theorem~\ref{thm:bound_bernoulli_process}]
	Using Lemma~\ref{lem:bernoulli_expectation} together with
	Theorem~\ref{thm:bound_the_metric} we obtain that for $T \subseteq \sqrt{s}
	B_{\ell^1}^N$, $\delta \in (0,1)$,
	$n_0 = \lceil \log_2\log_2(2N) + \log_2\log_2(sK^2/\delta) \rceil$ (thus, we obtain $2^{ \frac{n_0}{2}} \leq
	\sqrt{2}\sqrt{\log_2(2N) \log_2(sK^2/\delta)}$) and $\ell =\lceil \log_2(sK^2/\delta) \rceil
	\leq 2 \log_2(sK^2/\delta)$ the following bound holds:
	\begin{equation*}
		\begin{split}
		\ebb \sup_{f \in T} \Big| \frac{1}{m}
		\sum_{i=1}^m &|\inner{f}{X_i}|^2 \ep_i \Big|
		\leq (69+49\sqrt{2}) \, \delta \sum_{i=1}^m |\inner{f}{X_i}|^2
		+ (643+468\sqrt{2}) \, \delta m\\
		& +(280 + 200 \sqrt{2}) \, \sqrt{ \frac{s K^2 \log^2(sK^2/\delta) \log(eN)}{m} }
		\Big( \ebb \sup_{f \in T} \frac{1}{m} \sum_{i=1}^m |\inner{f}{X_i}|^2 \Big)^{1/2}
		 \; .
	\end{split}
	\end{equation*}
	This is the desired estimate.
\end{proof}

\subsection{Extension to Theorem~\ref{thm:weighted_main}}
In this section we discuss the changes of the proof of Theorem~\ref{prop:main} that are necessary in order to obtain its weighted version, Theorem~\ref{thm:weighted_main}. The
result in Section~\ref{sec:weighted_l1} can be obtained by the same argument as
presented in this section with minor modifications. For the convenience of the
reader this section discusses the necessary changes and their impact on the
argument.

As a first step let us note that the arguments stated up to the
subsection~\ref{subsec:estimating_b} are also valid in the context of
Theorem~\ref{thm:weighted_main} and can be applied without changes after recognizing
that under the assumption that $w_j \geq \norm{\inner{X_i}{e_j}}{L^\infty}$ for all
$j \in [N]$, then for each $f \in \cbb^N$ with $\norm{f}{w,1} \leq \sqrt{s}$ we have
\begin{equation}
	\label{eq:bound_inner_product}
	|\inner{f}{X_j}| \leq \norm{f}{w,1}
	\max_{j \in [N]} w_j^{-1} |X_j|
	\leq \norm{f}{w,1} \leq \sqrt{s}
	\;.
\end{equation}
Therefore,
the main content of this subsection is to show that we can construct an
admissible sequence $(A_n)_{n \geq n_0}$, which mimics the behaviour of the
sequence we constructed in Theorem~\ref{thm:bound_the_metric}. In order to
do this, let us start by introducing some notation.

The argument that we would like to adopt is concerned with a weighted version
of the $\ell^1$-ball. Recall from \eqref{eq:def/weighted} that the definition of
the weighted $\ell^p$-spaces gives the following definition for the weighted
$\ell^1$-norm. For a sequence of weights $w \in [1,\infty)^N$ we have
\begin{equation}
	\norm{f}{w,1} = \sum_{j=1}^N w_j |f_j|
	\label{eq:definition_wl1}.
\end{equation}
The unit ball of this norm is given by
\begin{equation}
	B_{\ell_w^1}^N := \{ f \in \cbb^N : \norm{f}{w,1} \leq 1\} \; .
	\label{eq:definition_wl1ball}
\end{equation}
Let us now indicate how the argument of Subsection~\ref{subsec:estimating_b}
has to be adjusted in order to cover the set $\{ (|\inner{f}{X_i}|^2)_{i \in [m]}
: f \in T\}$, where $T \subseteq \sqrt{s} B_{\ell_w^1}^N$.

The first point is to find a weak covering of $\sqrt{s} B_{\ell_w^1}^N$ in order
to get the argument started.
\begin{lem}
	Let $T \subset \sqrt{s} B_{\ell_w^1}^N$, $\delta \in (0,1)$ and $\rho>0$. Assume that
	$w \in [1,\infty)^N$ satisfies $w_j \geq \norm{\inner{X}{e_j}}{L^\infty}$. Then,
	for every realization of $X_1, \ldots, X_m$ we have
	\begin{equation*}
		\log \N^*\Big( T, \rho, \frac{4 \delta m}{s \log_2(s/\delta)} \Big)
		\leq \frac{2 \log_2(s/\delta) \log_2(2 N) s}{\rho^2}
	\end{equation*}
\end{lem}
\begin{proof}[Sketch of Proof]
The lemma is a staightforward adaption of Lemma~\ref{lem:maurey}.
First observe that every $x \in B_{\ell_w^1}^N$ is a convex combination of
$V= \{ \pm w_j \sqrt{s} e_j , \pm i w_j \sqrt{s} e_j : j \in [N]\}$, i.e.,
there are coefficients $\lambda_v \geq 0$ with $\sum_{v \in V} \lambda_v = 1$,
such that
\begin{equation*}
	x = \sum_{v \in V} \lambda_v v \; .
\end{equation*}
By arguing as in the proof of Lemma~\ref{lem:maurey} we can construct random
variables $Z_l$, for $l = 1, \ldots, L$, such that
$
\pr(Z_l = v) = \lambda_v
$. This sequence of random variables satisfies $\ebb \inner{Z_l}{X_i} =
\inner{x}{X_i}$ and for each $i \in [m]$,
\begin{equation*}
	\pr\Big( \Big| \sum_{l=1}^L \inner{Z_l}{X_i} - \inner{x}{X_i} \Big|
	\geq \rho \Big) \leq 2 \exp(-L \rho^2/ (\max_{l \in [L]} |\inner{X_i}{Z_l}|))
	\leq 2 \exp(-L \rho^2/s),
\end{equation*}
since $|\inner{X_i}{Z_l}| \leq \norm{Z_l}{w,1} \max_{j \in [N]} w_j^{-1}
|\inner{X_i}{e_j}| \leq \sqrt{s}$. The rest of the argument is identical to
the arguments in the proof of Lemma~\ref{lem:maurey}.
\end{proof}

\subsubsection{Adaption of Theorem~\ref{thm:bound_the_metric}}
Let us state a version of Theorem~\ref{thm:bound_the_metric}, which can be used in
the weighted context. 
The theorem and its proof only contain minor changes
compared to the result in Section~\ref{subsec:estimating_b}. 
\begin{thm}
	\label{thm:weighted_bound_the_metric}
	Let $T \subseteq \sqrt{s} B_{\ell_w^1}^N$ and let $\delta \in (0,1)$. Assume
	that for all $j \in [N]$ we have $w_j \geq \norm{\inner{X}{e_j}}{L^\infty}$.
	Then, for $n_0 = \lceil \log_2\log_2(eN) + \log_2 \log_2(s/\delta) \rceil$ and $\ell
	= \lceil \log_2(s / \delta) \rceil$ there is an admissible sequence $(A_n)_{n
	\geq n_0}$ for $T$ such that
	\begin{itemize}
			\item For all $f \in T$ we have
			\begin{equation}
				\norm{ |f_\X|^2 - \pi_{n_0 + \ell} |f_\X|^2 }{1}
				\leq 
				(69+49\sqrt{2})\delta \sum_{i=1}^m |\inner{f}{X_i}|^2
				+ (643+468\sqrt{2}) \delta \; .
			\end{equation}
			\item For all $f \in T$ we have
				\begin{equation}
					\sum_{n=n_0}^{n_0 + \ell - 1} 2^{\frac{n}{2}}
					\norm{\pi_{n+1} |f_\X|^2 - \pi_n |f_\X|^2}{2}
					\leq 
					(200 + 140\sqrt{2}) \cdot \sqrt{\ell s} 2^{\frac{n_0}{2}}
					\; \norm{(\inner{f}{X_i})_{i\in[m]}  }{2} \; .
				\end{equation}
	\end{itemize}
\end{thm}
Let us summarize the necessary changes in the argument presented in Section~\ref{subsec:estimating_b}.
As in the covering argument above, the following assumption is crucial for the deduction
of Theorem~\ref{thm:weighted_bound_the_metric}:
\begin{equation}
	\text{for all }j \in [N] \text{ we have }
	w_j \geq \norm{\inner{X}{e_j}}{L^\infty} \; .
	\label{eq:assumption_weights}
\end{equation}
The remainder is organized around the principles used in Section~\ref{subsec:estimating_b} and only contains remarks on the minor
changes that need to be applied in order to obtain Theorem~\ref{thm:weighted_bound_the_metric}.

As a general note in order to understand the adaptations we remark that the estimates in
Section~\ref{subsec:estimating_b} mostly depend on the inner product
$|\inner{f}{X_i}|$ for elements $f \in T \subseteq \sqrt{s} B_{\ell^1}^N$. However,
under the assumption \eqref{eq:assumption_weights} the inner product $|\inner{f}{X_i}|$
satisfies the same bounds as before.

\textit{The initial admissible sequence.} In oder to set up the initial admissible
sequence we consider the norm $\norm{f}{\X} := \max_{i \in [m]} |\inner{f}{X_i}|$
and observe that for $T \subseteq \sqrt{s} B_{\ell_w^1}^N$ we have the estimate
$\sup_{f \in T} \norm{f}{\X} \leq \sqrt{s}$ provided that \eqref{eq:assumption_weights}
is satisfied. Therefore, the initial sequence is defined as in \eqref{eq:definition_rho}
by setting
\begin{equation}
	\rho_n = \sqrt{s} 2^{-n/2} \quad \text{and} \quad M= \frac{4 \delta m}{s \log_2(s/\delta)} \; .
\end{equation}
The above lemma then garuantees the existence of an admissible sequence for $n= 0,
\cdots,\ell + k_0$. The necessary definitions in \eqref{eq:definition_differences}
and \eqref{eq:definition_admissible_set} can be adapted without changes.

\textit{$\ell^2$-stability.} Versions of Lemma~\ref{lem:bound_g} and
Lemma~\ref{lem:approximation_level_sets} for the weighted setting can be deduced
by the same arguments as in Subsection~\ref{subsec:estimating_b}.

\textit{The $\gamma_2$-boundedness.} A version of Lemma~\ref{lem:Euclidean_distance_bound}
can be deduced by exactly the same arguments. Further, the bound for the $\gamma_2$
functional in the proof of Theorem~\ref{thm:bound_the_metric} can be used line by
line by replacing $\sqrt{s} K$ by $\sup_{f \in T} \norm{f}{\X}$.

\textit{The $\ell^1$-approximation.} In order to deduce a version of
Lemma~\ref{lem:approximation} we can again use the argument presented in
Subsection~\ref{subsec:estimating_b}. This is possible, since all arguments in
the proof only depend on $\max_{i \in [m]} |\inner{f}{X_i}|^2$ and the differences
$\Big| |\inner{f}{X_i}| - |\inner{\widetilde{\pi}_{n_0 + k}(f)}{X_i}| \Big|
\1_{E_{n_0 + n}}(i)$. Further, by the choice of $\rho_n$,
Assumption~\eqref{eq:assumption_weights}, it follows
for $f \in \sqrt{s} B_{\ell_w^1}^N$ we have the bound $\norm{f}{\X} \leq \sqrt{s}$.

\section*{Acknowledgments}
S.B.\ acknowledges the Postdoctoral Training Centre in Stochastics of the Pacific Institute for the Mathematical Sciences (PIMS), the  Natural Sciences and Engineering Research Council of Canada (NSERC) through grant number 611675, the  Centre for Advanced Modelling Science (CADMOS), and the Faculty of Arts and Science of Concordia University for their financial support.
S.D., H.C.J., and H.R.\ acknowledge funding by the Deutsche Forschungsgemeinschaft (DFG, German Research Foundation) under SPP 1798 (COSIP - Compressed Sensing in Information Processing) through the project Quantized Compressive Spectrum Sensing.
\bibliographystyle{plain}
\bibliography{exposition_on_bos_sampling}

\end{document}